\documentclass[conference]{IEEEtran}

\usepackage{graphicx,fullpage,amsmath,amssymb,mathrsfs,color}
\usepackage{bold-extra}
\usepackage{times}
\usepackage{amsfonts,latexsym,graphics}
\usepackage{amsthm,amstext,url,paralist}
\usepackage{tikz}
\usetikzlibrary{arrows,automata,shapes,calc}
\usepackage[colorlinks=true,urlcolor=blue]{hyperref}
\usetikzlibrary{intersections}
\newtheorem{theorem}{Theorem}

\newtheorem{corollary}{Corollary}
\newtheorem{remark}{Remark}

\usepackage{graphicx}
\usepackage{caption}
\usepackage{subcaption}
\usepackage{pgfplots}
\pgfplotsset{width=10cm,compat=newest}

\usepackage{textcomp}
\usetikzlibrary{shapes,arrows}
\usepackage{cite}
\usepackage{pgfplots}

\pgfplotsset{compat=newest}

\allowdisplaybreaks[3]
\tikzset{
  sum/.style      = {draw, circle,inner sep=0.2mm, node distance = 2cm}, 
  input/.style    = {coordinate}, 
  output/.style   = {coordinate} 

}
\newcommand{\suma}{\Large$+$}
\newcommand{\dott}{\Large$.$}
\newcommand{\bx}{\mathbf{x}}
\newcommand{\bz}{\mathbf{z}}
\newcommand{\by}{\mathbf{y}}
\newcommand{\bs}{\mathbf{s}}

\newcommand{\bw}{\mathbf{w}}

\newcommand{\cF}{\cal{F}}
\newcommand{\cN}{\cal{N}}
\newcommand{\cA}{\cal{A}}

\newcommand{\cJ}{\cal{J}}
\newcommand{\cI}{\cal{I}}
\newcommand{\cB}{\cal{B}}
\newcommand{\cM}{\cal{M}}
\newcommand{\cS}{\cal{S}}
\newcommand{\cMI}{\cal{MI}}
\newcommand{\cOI}{\cal{OI}}

\makeatletter
\newcommand{\vast}{\bBigg@{4}}
\newcommand{\Vast}{\bBigg@{3.5}}
\makeatother

\begin{document}
\title{On the Sum Capacity of Many-to-one and One-to-many Gaussian Interference Channels }
\author{\IEEEauthorblockN{Abhiram Gnanasambandam, Ragini Chaluvadi, Srikrishna Bhashyam}
\IEEEauthorblockA{Department of Electrical Engineering,
Indian Institute of Technology Madras,
Chennai-600036, India \\
Email:abhiram.g94@gmail.com,\{ee14d404,skrishna\}@ee.iitm.ac.in}
}
\maketitle
\begin{abstract}
We obtain new sum capacity results for the Gaussian many-to-one and one-to-many interference channels in channel parameter regimes where the sum capacity was known only up to a constant gap. {\em Simple Han-Kobayashi (HK) schemes}, i.e., HK schemes with Gaussian signaling, no time-sharing, and no common-private power splitting, achieve sum capacity under the channel conditions for which the new results are obtained. To obtain sum capacity results, we show that genie-aided upper bounds match the achievable sum rate of  simple HK schemes under certain channel conditions.  
\end{abstract}

\allowdisplaybreaks

\section{Introduction}
The $K$-user Gaussian Interference channel (IC) has $K$ distinct
transmit-receive pairs that interfere with each other. The capacity
region or even the sum capacity are not known in general. The sum
capacity of the Gaussian IC is known under some channel conditions
\cite{rag1,ShaKraChe08,rag2,rag5,annvee09}. In \cite{rag1}, the
capacity region and sum capacity for the 2-user IC were determined
under strong interference conditions. In
\cite{ShaKraChe08,rag2,rag5,annvee09}, the sum capacity of the
$K$-user Gaussian IC was obtained under {\em noisy} interference
conditions. Under these conditions, Gaussian signaling and treating
interference as noise at each receiver achieves sum capacity. In
\cite{rag2}, the sum capacity of the 2-user Gaussian IC under mixed
interference conditions was also obtained.

The many-to-one Gaussian IC and one-to-many Gaussian IC are special
cases of the Gaussian IC where only one receiver experiences
interference or only one transmitter causes interference. Even for these simpler topologies, exact capacity results
are hard to obtain. The one-to-many IC and many-to-one IC were studied
in \cite{rag6,rag7,annvee09,RP10,RangaP,PraBhaCho14}. In \cite{rag6,rag7},
approximate capacity and degrees of freedom results are obtained for the many-to-one and one-to-many ICs. The sum capacity under {\em noisy} interference conditions is obtained for the many-to-one and one-to-many Gaussian ICs in \cite{annvee09,RP10}. The same results can also be obtained as a special case of the result in \cite{ShaKraChe08}. Recently, for the many-to-one Gaussian IC, channel conditions under which Gaussian signalling and a combination of treating interference as noise and interference decoding is sum rate optimal were obtained in \cite{RangaP}. In \cite{Tun11}, sum capacity was obtained for $K$-user Gaussian $Z$-like interference channels under some channel conditions. In both \cite{RangaP} and \cite{Tun11}, a {\em successive decoding} strategy where interference is decoded before decoding the desired signal is considered.
For the symmetric many-to-one IC, structured lattice codes were shown to achieve sum capacity under some strong interference conditions in \cite{zhu15}. Other special cases of the Gaussian IC, namely the cyclic IC and cascade IC were studied in \cite{cyclic,LiuErk11}.

In this paper, we obtain new sum capacity results for Gaussian many-to-one and one-to-many ICs. First, by careful Fourier-Motzkin elimination, we obtain the Han-Kobayashi (HK) achievable rate region for the $K$-user Gaussian many-to-one and one-to-many channels in simplified form, i.e., only in terms of the $K$ rates $R_1$, $R_2$, $\hdots$, $R_K$. Then, we focus on {\em simple} HK schemes with Gaussian signaling, no timesharing, and no common-private power splitting. We show that genie-aided sum capacity upper bounds {\em match} the achievable sum rates of simple HK schemes under some channel conditions.  We also discuss how the genie-aided bounds used in this paper differ from the bounds in \cite{Nam15,Nam15arx} for the K-user many-to-one Gaussian IC. Overall, we obtain new sum capacity results for a larger subset of possible channel conditions than currently known in exisiting literature in \cite{annvee09,RP10,RangaP,zhu15,Tun11}. In \cite{annvee09,RP10} only the case when all the interference is treated as noise was considered. In \cite{zhu15}, only the symmetric many-to-one IC was considered. In \cite{Tun11,RangaP}, only a successive decoding strategy was considered. Furthermore, the conditions under which sum capacity is achieved in \cite{Tun11} are not obtained explicitly in terms of the channel parameters. We allow joint decoding of the desired and interfering signals as well and obtain conditions explicitly in terms of the channel parameters. 
In the simple HK schemes considered in our paper, either the interference from a particular transmitter is decoded fully or gets treated as noise. For the many-to-one case, we consider schemes where $k$ out of $K$-1 interfering signals are decoded at receiver 1. For the one-to-many case, we consider schemes where $k$ out of $K$-1 receivers decode the interfering signal.

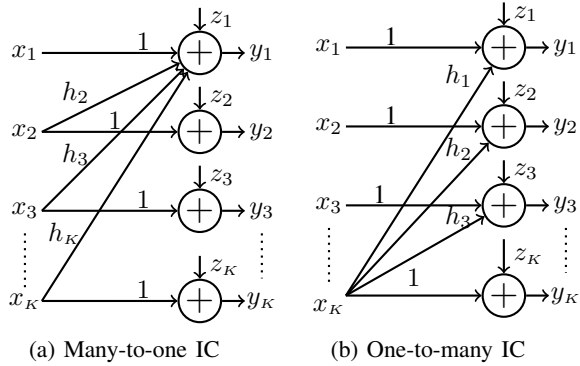
\begin{figure}[h!]
\begin{subfigure}[t]{0.2\textwidth}
\begin{tikzpicture}[scale=0.15,auto, thick, node distance=1cm]
\draw 
      node at (10,0) [sum, name=suma1] {\suma}
      node at (10,-7) [sum, name=suma2] {\suma}
      node at (10,-14) [sum, name=suma3] {\suma}
      node at (10,-22) [sum, name=suma4] {\suma};
\draw[->](-4,0) -- (suma1);
\draw[->](-4,-7) --(suma2);
\draw[->](-4,-14) -- (suma3);
\draw[->](-4,-22) --(suma4);
\draw[->](-4,-7) -- (suma1);
\draw[->](-4,-14) -- (suma1);
\draw[->](-4,-22) -- (suma1);

\draw[->] (10,4) -- node{$z_1$}(suma1);
\draw[->] (10,-3) -- node{$z_2$}(suma2);
\draw[->] (10,-10) -- node{$z_3$}(suma3);
\draw[->] (10,-18) -- node{$z_{\scriptscriptstyle{K}}$}(suma4);

\draw[->] (suma1) -- (14,0); 
\draw[->] (suma2) -- (14,-7);
\draw[->] (suma3) -- (14,-14);
\draw[->] (suma4) -- (14,-22);
\draw 

  node at (-1,-3.5) {$h_{2}$}
  node at (-1,-9) {$h_{3}$}
  node at (-2,-16) {$h_{\scriptscriptstyle{K}}$}
  node at (2.5,-6) {$1$}
  node at (5,1) {$1$}
  node at (5,-21) {$1$}
  node at (15.5,0) {$y_1$}
  node at (15.5,-7) {$y_2$}
  node at (15.5,-14) {$y_3$}
  node at (15.5,-22) {$y_{\scriptscriptstyle{K}}$}
  node at (-5.5,0) {$x_1$}
  node at (-5.5,-7) {$x_2$}
  node at (-5.5,-14) {$x_3$}
  node at (-5.5,-22) {$x_{\scriptscriptstyle{K}}$}
  node at (5,-13) {$1$};
  \draw[-,dotted] (15.5,-16) -- (15.5,-20);
    \draw[-,dotted] (-5.5,-16) -- (-5.5,-21);
\end{tikzpicture}
\caption{Many-to-one IC}
\label{fig:one-to-many}
\end{subfigure}\ \ \ \ \ \ 
\begin{subfigure}[t]{0.2\textwidth}
\centering
\begin{tikzpicture}[scale=0.15,auto, thick, node distance=2cm]
\draw 
      node at (12,0) [sum, name=suma1] {\suma}
      node at (12,-7) [sum, name=suma2] {\suma}
      node at (12,-14) [sum, name=suma3] {\suma}
      node at (12,-22) [sum, name=suma4] {\suma};
\draw[->](-2,0) -- (suma1);
\draw[->](-2,-7) --(suma2);
\draw[->](-2,-14) -- (suma3);
\draw[->](-2,-22) -- node {$1$}(suma4);
\draw[->](-2,-22) -- (suma1);
\draw[->](-2,-22) -- (suma2);
\draw[->](-2,-22) -- (suma3);

\draw[->] (12,4) -- node{$z_1$}(suma1);
\draw[->] (12,-3) -- node{$z_2$}(suma2);
\draw[->] (12,-10) -- node{$z_3$}(suma3);
\draw[->] (12,-17) -- node{$z_{\scriptscriptstyle{K}}$}(suma4);

\draw[->] (suma1) -- (16,0); 
\draw[->] (suma2) -- (16,-7);
\draw[->] (suma3) -- (16,-14);
\draw[->] (suma4) -- (16,-22);
\draw 
  node at (8,-2.5) {$h_{1}$}
  node at (8,-9) {$h_{2 }$}
  node at (8,-15) {$h_{3 }$}
  node at (2,-6) {1}
  node at (2,1) {1}
  node at (17.5,0) {$y_1$}
  node at (17.5,-7) {$y_2$}
  node at (17.5,-14) {$y_3$}
  node at (17.5,-22) {$y_{\scriptscriptstyle{K}}$}
  node at (-3.5,0) {$x_1$}
  node at (-3.5,-7) {$x_2$}
  node at (-3.5,-14) {$x_3$}
  node at (-3.5,-23) {$x_{\scriptscriptstyle{K}}$}
  node at (1,-13) {1};
  \draw[-,dotted] (17.5,-16) -- (17.5,-20);
    \draw[-,dotted] (-3.5,-16) -- (-3.5,-21);
\end{tikzpicture}
\caption{One-to-many IC}
\end{subfigure}
\caption{Channel Models in standard form}
\label{Model_fig}
\end{figure}
The channel models (in standard form) for the Gaussian many-to-one and one-to-many ICs are shown in Fig. \ref{Model_fig}. As an illustration of the new results in this paper, the channel conditions under which sum capacity results are obtained for the 3-user many-to-one and one-to-many ICs are shown in Figs. \ref{fig1} and \ref{fig2}. In the figures, the shaded regions represent the new regions where sum capacity is determined in this paper. 
\noindent
\begin{figure}

\centering
\resizebox{5.5cm}{5.5cm}{
\begin{tikzpicture}
\begin{axis}[xlabel={$|h_2|$},
    ylabel={$|h_3|$},
    xmin=0, xmax=4,
    ymin=0, ymax=4,
    xtick={0,1,2,3,4},
    ytick={0,1,2,3,4},
]

\addplot[ color=blue]
    coordinates {
    (1.732051e+00,0)(1.732109e+00,1.000000e-02)(1.732282e+00,2.000000e-02)(1.732570e+00,3.000000e-02)(1.732974e+00,4.000000e-02)(1.733494e+00,5.000000e-02)(1.734128e+00,6.000000e-02)(1.734878e+00,7.000000e-02)(1.735742e+00,8.000000e-02)(1.736721e+00,9.000000e-02)(1.737815e+00,1.000000e-01)(1.739023e+00,1.100000e-01)(1.740345e+00,1.200000e-01)(1.741781e+00,1.300000e-01)(1.743330e+00,1.400000e-01)(1.744993e+00,1.500000e-01)(1.746768e+00,1.600000e-01)(1.748657e+00,1.700000e-01)(1.750657e+00,1.800000e-01)(1.752769e+00,1.900000e-01)(1.754993e+00,2.000000e-01)(1.757328e+00,2.100000e-01)(1.759773e+00,2.200000e-01)(1.762328e+00,2.300000e-01)(1.764993e+00,2.400000e-01)(1.767767e+00,2.500000e-01)(1.770650e+00,2.600000e-01)(1.773640e+00,2.700000e-01)(1.776739e+00,2.800000e-01)(1.779944e+00,2.900000e-01)(1.783255e+00,3.000000e-01)(1.786673e+00,3.100000e-01)(1.790196e+00,3.200000e-01)(1.793823e+00,3.300000e-01)(1.797554e+00,3.400000e-01)(1.801388e+00,3.500000e-01)(1.805325e+00,3.600000e-01)(1.809365e+00,3.700000e-01)(1.813505e+00,3.800000e-01)(1.817746e+00,3.900000e-01)(1.822087e+00,4.000000e-01)(1.826527e+00,4.100000e-01)(1.831065e+00,4.200000e-01)(1.835702e+00,4.300000e-01)(1.840435e+00,4.400000e-01)(1.845264e+00,4.500000e-01)(1.850189e+00,4.600000e-01)(1.855209e+00,4.700000e-01)(1.860323e+00,4.800000e-01)(1.865529e+00,4.900000e-01)(1.870829e+00,5.000000e-01)(1.876220e+00,5.100000e-01)(1.881701e+00,5.200000e-01)(1.887273e+00,5.300000e-01)(1.892934e+00,5.400000e-01)(1.898684e+00,5.500000e-01)(1.904521e+00,5.600000e-01)(1.910445e+00,5.700000e-01)(1.916455e+00,5.800000e-01)(1.922550e+00,5.900000e-01)(1.928730e+00,6.000000e-01)(1.934994e+00,6.100000e-01)(1.941340e+00,6.200000e-01)(1.947768e+00,6.300000e-01)(1.954277e+00,6.400000e-01)(1.960867e+00,6.500000e-01)(1.967537e+00,6.600000e-01)(1.974285e+00,6.700000e-01)(1.981111e+00,6.800000e-01)(1.988014e+00,6.900000e-01)(1.994994e+00,7.000000e-01)(2.002049e+00,7.100000e-01)(2.009179e+00,7.200000e-01)(2.016383e+00,7.300000e-01)(2.023660e+00,7.400000e-01)(2.031010e+00,7.500000e-01)(2.038431e+00,7.600000e-01)(2.045923e+00,7.700000e-01)(2.053485e+00,7.800000e-01)(2.061116e+00,7.900000e-01)(2.068816e+00,8.000000e-01)(2.076584e+00,8.100000e-01)(2.084418e+00,8.200000e-01)(2.092319e+00,8.300000e-01)(2.100286e+00,8.400000e-01)(2.108317e+00,8.500000e-01)(2.116412e+00,8.600000e-01)(2.124571e+00,8.700000e-01)(2.132792e+00,8.800000e-01)(2.141074e+00,8.900000e-01)(2.149419e+00,9.000000e-01)(2.157823e+00,9.100000e-01)(2.166287e+00,9.200000e-01)(2.174810e+00,9.300000e-01)(2.183392e+00,9.400000e-01)(2.192031e+00,9.500000e-01)(2.200727e+00,9.600000e-01)(2.209480e+00,9.700000e-01)(2.218288e+00,9.800000e-01)(2.227151e+00,9.900000e-01)(2.236068e+00,1)
     (3,1) (6,1)
    }\closedcycle;

\addplot[
    color=blue]
    coordinates {   
    (0,1.732051e+00)(1.000000e-02,1.732109e+00)(2.000000e-02,1.732282e+00)(3.000000e-02,1.732570e+00)(4.000000e-02,1.732974e+00)(5.000000e-02,1.733494e+00)(6.000000e-02,1.734128e+00)(7.000000e-02,1.734878e+00)(8.000000e-02,1.735742e+00)(9.000000e-02,1.736721e+00)(1.000000e-01,1.737815e+00)(1.100000e-01,1.739023e+00)(1.200000e-01,1.740345e+00)(1.300000e-01,1.741781e+00)(1.400000e-01,1.743330e+00)(1.500000e-01,1.744993e+00)(1.600000e-01,1.746768e+00)(1.700000e-01,1.748657e+00)(1.800000e-01,1.750657e+00)(1.900000e-01,1.752769e+00)(2.000000e-01,1.754993e+00)(2.100000e-01,1.757328e+00)(2.200000e-01,1.759773e+00)(2.300000e-01,1.762328e+00)(2.400000e-01,1.764993e+00)(2.500000e-01,1.767767e+00)(2.600000e-01,1.770650e+00)(2.700000e-01,1.773640e+00)(2.800000e-01,1.776739e+00)(2.900000e-01,1.779944e+00)(3.000000e-01,1.783255e+00)(3.100000e-01,1.786673e+00)(3.200000e-01,1.790196e+00)(3.300000e-01,1.793823e+00)(3.400000e-01,1.797554e+00)(3.500000e-01,1.801388e+00)(3.600000e-01,1.805325e+00)(3.700000e-01,1.809365e+00)(3.800000e-01,1.813505e+00)(3.900000e-01,1.817746e+00)(4.000000e-01,1.822087e+00)(4.100000e-01,1.826527e+00)(4.200000e-01,1.831065e+00)(4.300000e-01,1.835702e+00)(4.400000e-01,1.840435e+00)(4.500000e-01,1.845264e+00)(4.600000e-01,1.850189e+00)(4.700000e-01,1.855209e+00)(4.800000e-01,1.860323e+00)(4.900000e-01,1.865529e+00)(5.000000e-01,1.870829e+00)(5.100000e-01,1.876220e+00)(5.200000e-01,1.881701e+00)(5.300000e-01,1.887273e+00)(5.400000e-01,1.892934e+00)(5.500000e-01,1.898684e+00)(5.600000e-01,1.904521e+00)(5.700000e-01,1.910445e+00)(5.800000e-01,1.916455e+00)(5.900000e-01,1.922550e+00)(6.000000e-01,1.928730e+00)(6.100000e-01,1.934994e+00)(6.200000e-01,1.941340e+00)(6.300000e-01,1.947768e+00)(6.400000e-01,1.954277e+00)(6.500000e-01,1.960867e+00)(6.600000e-01,1.967537e+00)(6.700000e-01,1.974285e+00)(6.800000e-01,1.981111e+00)(6.900000e-01,1.988014e+00)(7.000000e-01,1.994994e+00)(7.100000e-01,2.002049e+00)(7.200000e-01,2.009179e+00)(7.300000e-01,2.016383e+00)(7.400000e-01,2.023660e+00)(7.500000e-01,2.031010e+00)(7.600000e-01,2.038431e+00)(7.700000e-01,2.045923e+00)(7.800000e-01,2.053485e+00)(7.900000e-01,2.061116e+00)(8.000000e-01,2.068816e+00)(8.100000e-01,2.076584e+00)(8.200000e-01,2.084418e+00)(8.300000e-01,2.092319e+00)(8.400000e-01,2.100286e+00)(8.500000e-01,2.108317e+00)(8.600000e-01,2.116412e+00)(8.700000e-01,2.124571e+00)(8.800000e-01,2.132792e+00)(8.900000e-01,2.141074e+00)(9.000000e-01,2.149419e+00)(9.100000e-01,2.157823e+00)(9.200000e-01,2.166287e+00)(9.300000e-01,2.174810e+00)(9.400000e-01,2.183392e+00)(9.500000e-01,2.192031e+00)(9.600000e-01,2.200727e+00)(9.700000e-01,2.209480e+00)(9.800000e-01,2.218288e+00)(9.900000e-01,2.227151e+00)(1,2.236068e+00)  (1,3) (1,6) 
    };

    \addplot[
    fill=blue, 
                    fill opacity=0.2,
    color=blue ]
   coordinates {
    (0,1)(1.000000e-02,1.000250e+00)(2.000000e-02,1.001000e+00)(3.000000e-02,1.002251e+00)(4.000000e-02,1.004004e+00)(5.000000e-02,1.006259e+00)(6.000000e-02,1.009018e+00)(7.000000e-02,1.012283e+00)(8.000000e-02,1.016057e+00)(9.000000e-02,1.020341e+00)(1.000000e-01,1.025139e+00)(1.100000e-01,1.030453e+00)(1.200000e-01,1.036288e+00)(1.300000e-01,1.042648e+00)(1.400000e-01,1.049536e+00)(1.500000e-01,1.056958e+00)(1.600000e-01,1.064919e+00)(1.700000e-01,1.073425e+00)(1.800000e-01,1.082481e+00)(1.900000e-01,1.092093e+00)(2.000000e-01,1.102270e+00)(2.100000e-01,1.113019e+00)(2.200000e-01,1.124347e+00)(2.300000e-01,1.136262e+00)(2.400000e-01,1.148775e+00)(2.500000e-01,1.161895e+00)(2.600000e-01,1.175632e+00)(2.700000e-01,1.189996e+00)(2.800000e-01,1.205000e+00)(2.900000e-01,1.220656e+00)(3.000000e-01,1.236976e+00)(3.100000e-01,1.253975e+00)(3.200000e-01,1.271667e+00)(3.300000e-01,1.290068e+00)(3.400000e-01,1.309195e+00)(3.500000e-01,1.329064e+00)(3.600000e-01,1.349694e+00)(3.700000e-01,1.371105e+00)(3.800000e-01,1.393318e+00)(3.900000e-01,1.416354e+00)(4.000000e-01,1.440238e+00)(4.100000e-01,1.464994e+00)(4.200000e-01,1.490649e+00)(4.300000e-01,1.517231e+00)(4.400000e-01,1.544770e+00)(4.500000e-01,1.573298e+00)(4.600000e-01,1.602849e+00)(4.700000e-01,1.633460e+00)(4.800000e-01,1.665169e+00)(4.900000e-01,1.698018e+00)(5.000000e-01,1.732051e+00)(5.100000e-01,1.767316e+00)(5.200000e-01,1.803864e+00)(5.300000e-01,1.841750e+00)(5.400000e-01,1.881033e+00)(5.410000e-01,1.885041e+00)(5.411000e-01,1.885442e+00)(5.412000e-01,1.885844e+00)(5.413000e-01,1.886246e+00)(5.414000e-01,1.886648e+00)(5.415000e-01,1.887050e+00)(5.416000e-01,1.887452e+00)(5.417000e-01,1.887855e+00)(5.418000e-01,1.888257e+00)(5.419000e-01,1.888660e+00)(5.420000e-01,1.889063e+00)(5.421000e-01,1.889466e+00)(5.422000e-01,1.889869e+00)(5.423000e-01,1.890273e+00)(5.424000e-01,1.890676e+00)(5.425000e-01,1.891080e+00)(5.426000e-01,1.891483e+00)(5.427000e-01,1.891887e+00)(5.428000e-01,1.892291e+00)(5.429000e-01,1.892696e+00)(5.430000e-01,1.893100e+00)(5.431000e-01,1.893505e+00)(5.432000e-01,1.893909e+00)(5.433000e-01,1.894314e+00)(5.434000e-01,1.894719e+00) (5.400000e-01,1.892934e+00)
(5.200000e-01,1.881701e+00)(4.900000e-01,1.865529e+00)
(4.400000e-01,1.840435e+00)(3.600000e-01,1.805325e+00)
(3.200000e-01,1.790196e+00)(2.800000e-01,1.776739e+00)
(2.400000e-01,1.764993e+00)(2.000000e-01,1.754993e+00)
(1.600000e-01,1.746768e+00)(1.200000e-01,1.740345e+00)
(8.000000e-02,1.735742e+00)(2.000000e-02,1.732282e+00)
(0,1.732051e+00)
    
    }; 

    \addplot[fill=blue, 
                    fill opacity=0.2,
    color=blue  ]
   coordinates {
    (1,0)(1.000250e+00,1.000000e-02)(1.001000e+00,2.000000e-02)(1.002251e+00,3.000000e-02)(1.004004e+00,4.000000e-02)(1.006259e+00,5.000000e-02)(1.009018e+00,6.000000e-02)(1.012283e+00,7.000000e-02)(1.016057e+00,8.000000e-02)(1.020341e+00,9.000000e-02)(1.025139e+00,1.000000e-01)(1.030453e+00,1.100000e-01)(1.036288e+00,1.200000e-01)(1.042648e+00,1.300000e-01)(1.049536e+00,1.400000e-01)(1.056958e+00,1.500000e-01)(1.064919e+00,1.600000e-01)(1.073425e+00,1.700000e-01)(1.082481e+00,1.800000e-01)(1.092093e+00,1.900000e-01)(1.102270e+00,2.000000e-01)(1.113019e+00,2.100000e-01)(1.124347e+00,2.200000e-01)(1.136262e+00,2.300000e-01)(1.148775e+00,2.400000e-01)(1.161895e+00,2.500000e-01)(1.175632e+00,2.600000e-01)(1.189996e+00,2.700000e-01)(1.205000e+00,2.800000e-01)(1.220656e+00,2.900000e-01)(1.236976e+00,3.000000e-01)(1.253975e+00,3.100000e-01)(1.271667e+00,3.200000e-01)(1.290068e+00,3.300000e-01)(1.309195e+00,3.400000e-01)(1.329064e+00,3.500000e-01)(1.349694e+00,3.600000e-01)(1.371105e+00,3.700000e-01)(1.393318e+00,3.800000e-01)(1.416354e+00,3.900000e-01)(1.440238e+00,4.000000e-01)(1.464994e+00,4.100000e-01)(1.490649e+00,4.200000e-01)(1.517231e+00,4.300000e-01)(1.544770e+00,4.400000e-01)(1.573298e+00,4.500000e-01)(1.602849e+00,4.600000e-01)(1.633460e+00,4.700000e-01)(1.665169e+00,4.800000e-01)(1.698018e+00,4.900000e-01)(1.732051e+00,5.000000e-01)(1.767316e+00,5.100000e-01)(1.803864e+00,5.200000e-01)(1.841750e+00,5.300000e-01)(1.881033e+00,5.400000e-01)(1.885041e+00,5.410000e-01)(1.885442e+00,5.411000e-01)(1.885844e+00,5.412000e-01)(1.886246e+00,5.413000e-01)(1.886648e+00,5.414000e-01)(1.887050e+00,5.415000e-01)(1.887452e+00,5.416000e-01)(1.887855e+00,5.417000e-01)(1.888257e+00,5.418000e-01)(1.888660e+00,5.419000e-01)(1.889063e+00,5.420000e-01)(1.889466e+00,5.421000e-01)(1.889869e+00,5.422000e-01)(1.890273e+00,5.423000e-01)(1.890676e+00,5.424000e-01)(1.891080e+00,5.425000e-01)(1.891483e+00,5.426000e-01)(1.891887e+00,5.427000e-01)(1.892291e+00,5.428000e-01)(1.892696e+00,5.429000e-01)(1.893100e+00,5.430000e-01)(1.893505e+00,5.431000e-01)(1.893909e+00,5.432000e-01)(1.894314e+00,5.433000e-01)(1.894719e+00,5.434000e-01)(1.887273e+00,5.300000e-01)
(1.855209e+00,4.700000e-01)(1.817746e+00,3.900000e-01)
(1.786673e+00,3.100000e-01)(1.767767e+00,2.500000e-01)
(1.752769e+00,1.900000e-01)(1.741781e+00,1.300000e-01)
(1.741781e+00,1.300000e-01)(1.732570e+00,3.000000e-02)
(1.732051e+00,0)

    };   
 \addplot[color=black]   
coordinates{
        (0,1)(1.000000e-02,9.999500e-01)(2.000000e-02,9.998000e-01)(3.000000e-02,9.995499e-01)(4.000000e-02,9.991997e-01)(5.000000e-02,9.987492e-01)(6.000000e-02,9.981984e-01)(7.000000e-02,9.975470e-01)(8.000000e-02,9.967949e-01)(9.000000e-02,9.959418e-01)(1.000000e-01,9.949874e-01)(1.100000e-01,9.939316e-01)(1.200000e-01,9.927739e-01)(1.300000e-01,9.915140e-01)(1.400000e-01,9.901515e-01)(1.500000e-01,9.886860e-01)(1.600000e-01,9.871170e-01)(1.700000e-01,9.854441e-01)(1.800000e-01,9.836666e-01)(1.900000e-01,9.817841e-01)(2.000000e-01,9.797959e-01)(2.100000e-01,9.777014e-01)(2.200000e-01,9.754999e-01)(2.300000e-01,9.731906e-01)(2.400000e-01,9.707729e-01)(2.500000e-01,9.682458e-01)(2.600000e-01,9.656086e-01)(2.700000e-01,9.628603e-01)(2.800000e-01,9.600000e-01)(2.900000e-01,9.570266e-01)(3.000000e-01,9.539392e-01)(3.100000e-01,9.507366e-01)(3.200000e-01,9.474175e-01)(3.300000e-01,9.439809e-01)(3.400000e-01,9.404254e-01)(3.500000e-01,9.367497e-01)(3.600000e-01,9.329523e-01)(3.700000e-01,9.290318e-01)(3.800000e-01,9.249865e-01)(3.900000e-01,9.208149e-01)(4.000000e-01,9.165151e-01)(4.100000e-01,9.120855e-01)(4.200000e-01,9.075241e-01)(4.300000e-01,9.028289e-01)(4.400000e-01,8.979978e-01)(4.500000e-01,8.930286e-01)(4.600000e-01,8.879189e-01)(4.700000e-01,8.826664e-01)(4.800000e-01,8.772685e-01)(4.900000e-01,8.717224e-01)(5.000000e-01,8.660254e-01)(5.100000e-01,8.601744e-01)(5.200000e-01,8.541663e-01)(5.300000e-01,8.479976e-01)(5.400000e-01,8.416650e-01)(5.500000e-01,8.351647e-01)(5.600000e-01,8.284926e-01)(5.700000e-01,8.216447e-01)(5.800000e-01,8.146165e-01)(5.900000e-01,8.074032e-01)(6.000000e-01,8.000000e-01)(6.100000e-01,7.924014e-01)(6.200000e-01,7.846018e-01)(6.300000e-01,7.765951e-01)(6.400000e-01,7.683749e-01)(6.500000e-01,7.599342e-01)(6.600000e-01,7.512656e-01)(6.700000e-01,7.423611e-01)(6.800000e-01,7.332121e-01)(6.900000e-01,7.238094e-01)(7.000000e-01,7.141428e-01)(7.100000e-01,7.042017e-01)(7.200000e-01,6.939741e-01)(7.300000e-01,6.834471e-01)(7.400000e-01,6.726069e-01)(7.500000e-01,6.614378e-01)(7.600000e-01,6.499231e-01)(7.700000e-01,6.380439e-01)(7.800000e-01,6.257795e-01)(7.900000e-01,6.131068e-01)(8.000000e-01,6.000000e-01)(8.100000e-01,5.864299e-01)(8.200000e-01,5.723635e-01)(8.300000e-01,5.577634e-01)(8.400000e-01,5.425864e-01)(8.500000e-01,5.267827e-01)(8.600000e-01,5.102940e-01)(8.700000e-01,4.930517e-01)(8.800000e-01,4.749737e-01)(8.900000e-01,4.559605e-01)(9.000000e-01,4.358899e-01)(9.100000e-01,4.146082e-01)(9.200000e-01,3.919184e-01)(9.300000e-01,3.675595e-01)(9.400000e-01,3.411744e-01)(9.500000e-01,3.122499e-01)(9.600000e-01,2.800000e-01)(9.700000e-01,2.431049e-01)(9.800000e-01,1.989975e-01)(9.900000e-01,1.410674e-01)(1,0)
};  
\addplot[color=green]   
coordinates{
(1.732051e+00,6)(1.732051e+00,3)(1.742051e+00,2.994204e+00)(1.752051e+00,2.988364e+00)(1.762051e+00,2.982478e+00)(1.772051e+00,2.976548e+00)(1.782051e+00,2.970571e+00)(1.792051e+00,2.964550e+00)(1.802051e+00,2.958482e+00)(1.812051e+00,2.952367e+00)(1.822051e+00,2.946206e+00)(1.832051e+00,2.939998e+00)(1.842051e+00,2.933743e+00)(1.852051e+00,2.927440e+00)(1.862051e+00,2.921090e+00)(1.872051e+00,2.914691e+00)(1.882051e+00,2.908244e+00)(1.892051e+00,2.901748e+00)(1.902051e+00,2.895203e+00)(1.912051e+00,2.888609e+00)(1.922051e+00,2.881965e+00)(1.932051e+00,2.875270e+00)(1.942051e+00,2.868526e+00)(1.952051e+00,2.861730e+00)(1.962051e+00,2.854883e+00)(1.972051e+00,2.847984e+00)(1.982051e+00,2.841034e+00)(1.992051e+00,2.834031e+00)(2.002051e+00,2.826976e+00)(2.012051e+00,2.819867e+00)(2.022051e+00,2.812705e+00)(2.032051e+00,2.805489e+00)(2.042051e+00,2.798219e+00)(2.052051e+00,2.790894e+00)(2.062051e+00,2.783513e+00)(2.072051e+00,2.776077e+00)(2.082051e+00,2.768585e+00)(2.092051e+00,2.761037e+00)(2.102051e+00,2.753431e+00)(2.112051e+00,2.745768e+00)(2.122051e+00,2.738047e+00)(2.132051e+00,2.730267e+00)(2.142051e+00,2.722429e+00)(2.152051e+00,2.714531e+00)(2.162051e+00,2.706573e+00)(2.172051e+00,2.698554e+00)(2.182051e+00,2.690475e+00)(2.192051e+00,2.682334e+00)(2.202051e+00,2.674130e+00)(2.212051e+00,2.665864e+00)(2.222051e+00,2.657535e+00)(2.232051e+00,2.649141e+00)(2.242051e+00,2.640683e+00)(2.252051e+00,2.632160e+00)(2.262051e+00,2.623571e+00)(2.272051e+00,2.614916e+00)(2.282051e+00,2.606193e+00)(2.292051e+00,2.597403e+00)(2.302051e+00,2.588544e+00)(2.312051e+00,2.579616e+00)(2.322051e+00,2.570619e+00)(2.332051e+00,2.561550e+00)(2.342051e+00,2.552410e+00)(2.352051e+00,2.543198e+00)(2.362051e+00,2.533913e+00)(2.372051e+00,2.524554e+00)(2.382051e+00,2.515121e+00)(2.392051e+00,2.505612e+00)(2.402051e+00,2.496027e+00)(2.412051e+00,2.486365e+00)(2.422051e+00,2.476625e+00)(2.432051e+00,2.466805e+00)(2.442051e+00,2.456906e+00)(2.452051e+00,2.446926e+00)(2.462051e+00,2.436864e+00)(2.472051e+00,2.426719e+00)(2.482051e+00,2.416490e+00)(2.492051e+00,2.406176e+00)(2.502051e+00,2.395776e+00)(2.512051e+00,2.385288e+00)(2.522051e+00,2.374713e+00)(2.532051e+00,2.364047e+00)(2.542051e+00,2.353291e+00)(2.552051e+00,2.342442e+00)(2.562051e+00,2.331501e+00)(2.572051e+00,2.320464e+00)(2.582051e+00,2.309332e+00)(2.592051e+00,2.298102e+00)(2.602051e+00,2.286773e+00)(2.612051e+00,2.275344e+00)(2.622051e+00,2.263813e+00)(2.632051e+00,2.252179e+00)(2.642051e+00,2.240439e+00)(2.652051e+00,2.228593e+00)(2.662051e+00,2.216638e+00)(2.672051e+00,2.204574e+00)(2.682051e+00,2.192397e+00)(2.692051e+00,2.180106e+00)(2.702051e+00,2.167700e+00)(2.712051e+00,2.155175e+00)(2.722051e+00,2.142531e+00)(2.732051e+00,2.129765e+00)(2.742051e+00,2.116874e+00)(2.752051e+00,2.103857e+00)(2.762051e+00,2.090712e+00)(2.772051e+00,2.077435e+00)(2.782051e+00,2.064024e+00)(2.792051e+00,2.050476e+00)(2.802051e+00,2.036789e+00)(2.812051e+00,2.022961e+00)(2.822051e+00,2.008987e+00)(2.832051e+00,1.994865e+00)(2.842051e+00,1.980593e+00)(2.852051e+00,1.966165e+00)(2.862051e+00,1.951580e+00)(2.872051e+00,1.936834e+00)(2.882051e+00,1.921922e+00)(2.892051e+00,1.906841e+00)(2.902051e+00,1.891587e+00)(2.912051e+00,1.876156e+00)(2.922051e+00,1.860543e+00)(2.932051e+00,1.844743e+00)(2.942051e+00,1.828753e+00)(2.952051e+00,1.812566e+00)(2.962051e+00,1.796178e+00)(2.972051e+00,1.779583e+00)(2.982051e+00,1.762774e+00)(2.992051e+00,1.745747e+00)(3,1.732051e+00)(6,1.732051e+00)
} 	;
\addplot[color=black, fill=black, fill opacity = 0.3]   
coordinates{(1,4)
(1,2.236068e+00)(1.010000e+00,2.245039e+00)(1.020000e+00,2.254063e+00)(1.030000e+00,2.263139e+00)(1.040000e+00,2.272268e+00)(1.050000e+00,2.281447e+00)(1.060000e+00,2.290677e+00)(1.070000e+00,2.299957e+00)(1.080000e+00,2.309286e+00)(1.090000e+00,2.318663e+00)(1.100000e+00,2.328089e+00)(1.110000e+00,2.337563e+00)(1.120000e+00,2.347083e+00)(1.130000e+00,2.356650e+00)(1.140000e+00,2.366263e+00)(1.150000e+00,2.375921e+00)(1.160000e+00,2.385624e+00)(1.170000e+00,2.395371e+00)(1.180000e+00,2.405161e+00)(1.190000e+00,2.414995e+00)(1.200000e+00,2.424871e+00)(1.210000e+00,2.434790e+00)(1.220000e+00,2.444749e+00)(1.230000e+00,2.454750e+00)(1.240000e+00,2.464792e+00)(1.250000e+00,2.474874e+00)(1.260000e+00,2.484995e+00)(1.270000e+00,2.495155e+00)(1.280000e+00,2.505354e+00)(1.290000e+00,2.515591e+00)(1.300000e+00,2.525866e+00)(1.310000e+00,2.536178e+00)(1.320000e+00,2.546527e+00)(1.330000e+00,2.556912e+00)(1.340000e+00,2.567333e+00)(1.350000e+00,2.577790e+00)(1.360000e+00,2.588281e+00)(1.370000e+00,2.598807e+00)(1.380000e+00,2.609368e+00)(1.390000e+00,2.619962e+00)(1.400000e+00,2.630589e+00)(1.410000e+00,2.641250e+00)(1.420000e+00,2.651943e+00)(1.430000e+00,2.662668e+00)(1.440000e+00,2.673425e+00)(1.450000e+00,2.684213e+00)(1.460000e+00,2.695032e+00)(1.470000e+00,2.705882e+00)(1.480000e+00,2.716763e+00)(1.490000e+00,2.727673e+00)(1.500000e+00,2.738613e+00)(1.510000e+00,2.749582e+00)(1.520000e+00,2.760580e+00)(1.530000e+00,2.771606e+00)(1.540000e+00,2.782661e+00)(1.550000e+00,2.793743e+00)(1.560000e+00,2.804853e+00)(1.570000e+00,2.815990e+00)(1.580000e+00,2.827154e+00)(1.590000e+00,2.838345e+00)(1.600000e+00,2.849561e+00)(1.610000e+00,2.860804e+00)(1.620000e+00,2.872072e+00)(1.630000e+00,2.883366e+00)(1.640000e+00,2.894685e+00)(1.650000e+00,2.906028e+00)(1.660000e+00,2.917396e+00)(1.670000e+00,2.928788e+00)(1.680000e+00,2.940204e+00)(1.690000e+00,2.951644e+00)(1.700000e+00,2.963106e+00)(1.710000e+00,2.974592e+00)(1.720000e+00,2.986101e+00)(1.730000e+00,2.997632e+00)(1.732051e+00,3)(1.732051e+00,4) 
};
\addplot[color=black, fill=black, fill opacity = 0.3]   
coordinates{(4,1)
(2.236068e+00,1)(2.245039e+00,1.010000e+00)(2.254063e+00,1.020000e+00)(2.263139e+00,1.030000e+00)(2.272268e+00,1.040000e+00)(2.281447e+00,1.050000e+00)(2.290677e+00,1.060000e+00)(2.299957e+00,1.070000e+00)(2.309286e+00,1.080000e+00)(2.318663e+00,1.090000e+00)(2.328089e+00,1.100000e+00)(2.337563e+00,1.110000e+00)(2.347083e+00,1.120000e+00)(2.356650e+00,1.130000e+00)(2.366263e+00,1.140000e+00)(2.375921e+00,1.150000e+00)(2.385624e+00,1.160000e+00)(2.395371e+00,1.170000e+00)(2.405161e+00,1.180000e+00)(2.414995e+00,1.190000e+00)(2.424871e+00,1.200000e+00)(2.434790e+00,1.210000e+00)(2.444749e+00,1.220000e+00)(2.454750e+00,1.230000e+00)(2.464792e+00,1.240000e+00)(2.474874e+00,1.250000e+00)(2.484995e+00,1.260000e+00)(2.495155e+00,1.270000e+00)(2.505354e+00,1.280000e+00)(2.515591e+00,1.290000e+00)(2.525866e+00,1.300000e+00)(2.536178e+00,1.310000e+00)(2.546527e+00,1.320000e+00)(2.556912e+00,1.330000e+00)(2.567333e+00,1.340000e+00)(2.577790e+00,1.350000e+00)(2.588281e+00,1.360000e+00)(2.598807e+00,1.370000e+00)(2.609368e+00,1.380000e+00)(2.619962e+00,1.390000e+00)(2.630589e+00,1.400000e+00)(2.641250e+00,1.410000e+00)(2.651943e+00,1.420000e+00)(2.662668e+00,1.430000e+00)(2.673425e+00,1.440000e+00)(2.684213e+00,1.450000e+00)(2.695032e+00,1.460000e+00)(2.705882e+00,1.470000e+00)(2.716763e+00,1.480000e+00)(2.727673e+00,1.490000e+00)(2.738613e+00,1.500000e+00)(2.749582e+00,1.510000e+00)(2.760580e+00,1.520000e+00)(2.771606e+00,1.530000e+00)(2.782661e+00,1.540000e+00)(2.793743e+00,1.550000e+00)(2.804853e+00,1.560000e+00)(2.815990e+00,1.570000e+00)(2.827154e+00,1.580000e+00)(2.838345e+00,1.590000e+00)(2.849561e+00,1.600000e+00)(2.860804e+00,1.610000e+00)(2.872072e+00,1.620000e+00)(2.883366e+00,1.630000e+00)(2.894685e+00,1.640000e+00)(2.906028e+00,1.650000e+00)(2.917396e+00,1.660000e+00)(2.928788e+00,1.670000e+00)(2.940204e+00,1.680000e+00)(2.951644e+00,1.690000e+00)(2.963106e+00,1.700000e+00)(2.974592e+00,1.710000e+00)(2.986101e+00,1.720000e+00)(2.997632e+00,1.730000e+00)(3,1.732051e+00)(4,1.732051e+00)
};
\addplot[color = red,fill=red, fill opacity = 0.2]   
coordinates{
(2.6,4.0645)(2.5,3.9370)(2.4,3.8105)(2.3,3.6851)(2.2,3.5609)(2.1,3.4380)(2,3.3166)(1.9,3.1969)(1.8,3.0790)(1.732051e+00,3)(1.782051e+00,2.970571e+00)(1.832051e+00,2.939998e+00)(1.882051e+00,2.908244e+00)(1.932051e+00,2.875270e+00)(1.982051e+00,2.841034e+00)(2.032051e+00,2.805489e+00)(2.082051e+00,2.768585e+00)(2.132051e+00,2.730267e+00)(2.182051e+00,2.690475e+00)(2.232051e+00,2.649141e+00)(2.282051e+00,2.606193e+00)(2.332051e+00,2.561550e+00)(2.382051e+00,2.515121e+00)(2.432051e+00,2.466805e+00)(2.482051e+00,2.416490e+00)(2.532051e+00,2.364047e+00)(2.582051e+00,2.309332e+00)(2.632051e+00,2.252179e+00)(2.682051e+00,2.192397e+00)(2.732051e+00,2.129765e+00)(2.782051e+00,2.064024e+00)(2.832051e+00,1.994865e+00)(2.882051e+00,1.921922e+00)(2.932051e+00,1.844743e+00)(2.982051e+00,1.762774e+00)(2.992051e+00,1.745747e+00)(3,1.732051e+00) 
(3.0790,1.8)(3.1969,1.9)(3.3166,2)(3.4380,2.1)(3.5609,2.2)(3.6851,2.3)(3.8105,2.4)(3.9370,2.5)(4.0645,2.6)(4.0645,4.0645)
};
\addplot[color = brown , ultra thick]   
coordinates{(2.12,2.12)(4,4)
};
 \node[black] at (axis cs:3.5,0.5){${\cMI}2_{0}$}; 
  \node[black] at (axis cs:0.5,3.5){${\cMI}2_{0}$}; 
   \node[black] at (axis cs:2,3.7){${\cMI}3_{0}$}; 
   \node[black] at (axis cs:3.7,2){${\cMI}3_{0}$};
   \node[black] at (axis cs:3,3.5){${\cMI}3_{0}$}; 
   \node[black] at (axis cs:2.3,1.7){\cite{zhu15}};  
    \node[black] at (axis cs:0.4,0.4){${\cMI}1_0$}; 
    \node[black] at (axis cs:0.2,1.6){${\cMI}2_1$}; 
    \node[black] at (axis cs:1.5,0.25){${\cMI}2_1$}; 
    \node[black] at (axis cs:3.5,1.4){${\cMI}3_{1}$};
    \node[black] at (axis cs:1.4,3.5){${\cMI}3_{1}$};  
     \node[anchor=west] (source) at (axis cs:2.2,2.3){};
       \node (destination) at (axis cs:2.3,1.8){};
       \draw[->](source)--(destination);
\end{axis}

\end{tikzpicture}
}
\caption{\label{fig1} Channel conditions where sum capacity is obtained for the $3$-user many-to-one IC, $P_1=P_2=P_3=2$.}
\end{figure}
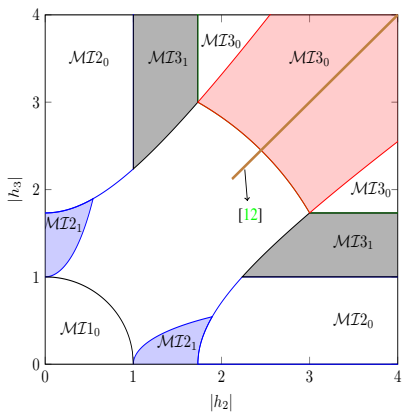
\begin{figure}
\label{fig:manytoone}\resizebox{6cm}{6cm}{
\begin{tikzpicture}
\begin{axis}[xlabel={$|h_1|$},
    ylabel={$|h_2|$},
    xmin=0, xmax=4,
    ymin=0, ymax=4,
    xtick={0,0.5,1,1.5,2,2.5,3,3.5,4},
    ytick={0,0.5,1,1.5,2,2.5,3,3.5,4},
]
\addplot[
    color=blue,fill=blue,opacity =0.2]
    coordinates {
    (1.414,0) (1.414,1) (3,1) (4,1) (4,0)
    };
\addplot[
    color=red,fill = red,opacity =0.2 ]
   coordinates {
    (4,1.414) (3,1.414) (1.414,1.414) (1.414,4)(4,4)
    };    
\addplot[
    color=black,fill = black,opacity =0.3  ]
   coordinates {
    (1.414,1.414)(1,1)(4,1)(4,1.414)
    };
\addplot[
    color=black,fill = black,opacity =0.3  ]
   coordinates {
    (1.414,1.414) (1,1)(1,4)(1.414,4)
    };
\addplot[
    color=blue,fill = blue,opacity = 0.2]
    coordinates {
    (0,1.414) (1,1.414) (1,3) (1,4)(0,4)
    };    
 \addplot[color=black]   
coordinates{
        ( 0,   1.0000)     (0.0918,0.9834)    (0.1954  , 0.9289) (0.2000  ,  0.9258 ) 
 ( 0.2045 ,   0.9227)    (0.2089  ,  0.9196 )    (0.3277 ,0.8216)   (0.3901   , 0.7629 ) (0.3930    ,0.7602 )(0.3958  ,  0.7574 )(0.3986  ,  0.7547) (0.4015    ,0.7520  )  (0.4043  ,  0.7492  )(0.4071  ,  0.7465 )   (0.4098   , 0.7438 )  
(0.4880 ,   0.6667)(0.4906 ,   0.6641)   
(0.4984 ,   0.6562)(0.5010 ,   0.6536)(0.5037 ,   0.6510 ) (0.5063 ,   0.6484)    
(0.5089 ,   0.6458 )
(0.5400 ,   0.6147)(0.5812 ,   0.5735)    (0.5838 ,   0.5709)(0.5864 ,   0.5683)       
(0.6328  ,  0.5219)(0.6354 ,   0.5193 )(0.6380 ,   0.5167 )(0.6406 ,   0.5141 )  
(0.6432 ,   0.5115) (0.6458 ,   0.5089 )(0.6484 ,   0.5063 )(0.6510 ,   0.5037 )  
(0.6536 ,   0.5010)(0.6562 ,   0.4984 )(0.6588 ,   0.4958 )(0.6614 ,   0.4932)
( 0.6641 ,   0.4906 )(0.6667 ,   0.4880  )(0.6693 ,   0.4853 )(0.6719 ,   0.4827 )   
(0.6745 ,   0.4801 )(0.6772 ,   0.4774 )(0.6798 ,   0.4748 )(0.6824 ,   0.472)   
(0.6850 ,   0.4695)(0.6877 ,   0.4668 )(0.6903 ,   0.4642 )(0.7656 ,   0.3873 ) (0.7684 ,   0.3844)(0.7711 ,  0.3816 )(0.7739 , 0.3787)(0.7767 , 0.3758) (0.8830 ,   0.2568)    (0.8860 ,   0.2531) (0.8981 ,   0.2378) (0.9012 ,   0.2338 ) (0.9042 ,   0.2298)(0.9073 ,   0.2258 )  
(0.9104 ,   0.2216)    (0.9134 ,   0.2175)  (0.9165 ,   0.2132) (0.9196 ,   0.2089)    
(0.9227 ,   0.2045)(0.9258 ,   0.2000) (0.9289 ,   0.1954)    (0.9321 ,   0.1908)    
(0.9352 ,   0.1860)    (0.9384 ,   0.1811) (0.9415 ,   0.1762)(0.9447 ,   0.1711 )   
(0.9478 ,   0.1658) (0.9510 ,   0.1604)(0.9542 ,   0.1548) (0.9574,    0.1491)    
(0.9606 ,   0.1431) (0.9639 ,   0.1369 )   (0.9671 ,   0.1304 ) (0.9703 ,   0.1236 )   
(0.9736 ,   0.1164 )   (0.9769 ,   0.1088) (0.9801 ,   0.1007 ) (0.9834 ,   0.0918)    
(0.9867 ,   0.0821) (0.9900 ,   0.0710) (0.9933 ,   0.0579 )  (0.9967 ,   0.0409 )   
(1.0000 , 0)
};  
 \node[black] at (axis cs:2.5,0.5){${\cOI}_1$}; 
  \node[black] at (axis cs:0.5,2.5){${\cOI}_1$}; 
   \node[black] at (axis cs:2.5,2.5){${\cOI}_2$}; 
    \node[black] at (axis cs:0.3,0.3){${\cOI}_0$}; 
    \node[black] at (axis cs:2.5,1.25){${\cOI}_{2_{1}}$};
    \node[black] at (axis cs:1.25,2.5){${\cOI}_{2_{1}}$};
\end{axis}
\end{tikzpicture}
}
\caption{\label{fig2} Channel conditions where sum capacity is obtained for the $3$-user one-to-many IC, $P_1=P_2=P_3=1$.}
\label{fig:MISO}
\end{figure}
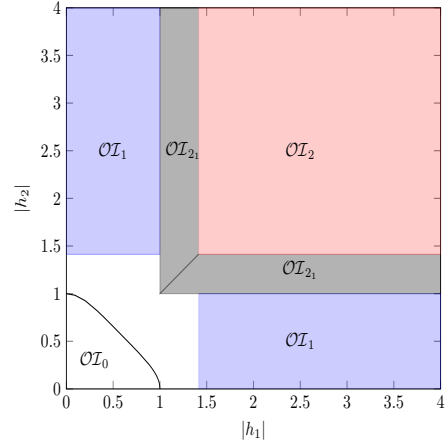
\section{Channel Models in standard form}
The received signals in the Gaussian many-to-one IC in standard form are given by:
\begin{eqnarray}
y_{1} = x_{1} + \underset{j=2}{\overset{K}{\sum}} h_{i}x_{i} + z_{1} \label{m21model1}\\
y_{i} = x_{i} + z_{i}, i = 2, 3, ..., K, \label{m21model2}
\end{eqnarray}
\noindent
where $x_{i}$ is transmitted from transmitter $i$, $z_i \sim \mathcal{N}(0,1)$ for each $i$. 
The average power constraint at transmitter $i$ is $P_i$. Similarly, the received signals in the Gaussian one-to-many IC in standard form are given by:
\begin{IEEEeqnarray}{lcr}
y_i =x_i+h_i x_{\scriptscriptstyle{K}}+z_i, \ i=1, 2, 3, \cdots, K - 1 \label{model1}\\
y_{\scriptscriptstyle{K}} =x_{\scriptscriptstyle{K}}+z_{\scriptscriptstyle{K}}. \label{model2}
\end{IEEEeqnarray} 

\section{Achievable rate region for Han-Kobayashi (HK) scheme in simplified form}
\subsection{Many-to-one IC}
Let $W_{i}$ be the message at transmitter $i$. For each $i = 2, 3,
..., K$, the message is split into two parts
$W_{i}=\{W_{i0},W_{i1}\}$, where $W_{i0}$ is common message that gets
decoded at receiver $i$ and also at receiver 1, and $W_{i1}$ is the private
message that gets decoded only at receiver $i$. The HK achievable rate region in simplified form in the Theorem below is stated for the discrete memoryless channel, and can be readily extended to the Gaussian many-to-one IC with average power constraints using standard approaches \cite{HanKob81,ElgKim11}.

\begin{theorem}\label{HKtheorem}
For the discrete memoryless $K$-user many-to-one IC, the HK achievable rate region is given by the set of all $(R_{1},R_{2},\hdots, R_{K})$ that satisfy: 
\begin{multline} \label{eqn:theorem_result_1}
R_{1} + \underset{j \in \cN}\sum R_{j} \leq \underset{j \in \cN} \sum I(X_{j};Y_{j}|Q, U_{j}) \\ + I(U_{\cN}X_{1};Y_{1}|U_{\cF-\cN}, Q) , \forall \cN \subseteq \cF
\end{multline}
\begin{equation}
\label{eqn:theorem_result_2}
R_{i} \leq I(X_{i};Y_{i}|Q), i\in [2:K]
\end{equation}
where $U_{\cA} = \{U_{i},i\in \cA\}$, ${\cF} = \{2, 3, ..., K\}$ and 
$(Q,U_{2},U_{3},\hdots, U_{\scriptscriptstyle{K}},X_{1},X_{2}, \hdots
X_{\scriptscriptstyle{K}})$ is distributed as
\begin{IEEEeqnarray}{lcr}
\label{input_dist}
\nonumber p(q,u_{2},\hdots,u{\scriptscriptstyle{K}},x_{1},\hdots, x_{\scriptscriptstyle{K}})\\
\nonumber = p(q)p(x_{1}|q)\underset{i=2}{\overset{K}{\prod}}(p(u_{i}|q)p(x_{i}|u_{i},q).
\end{IEEEeqnarray}
\end{theorem} 
\begin{proof}
See Appendix \ref{RateMotz}.
\end{proof}

\begin{corollary} \label{HKsumm21}
The achievable sum rate $S$ for a discrete memoryless many-to-one IC satisfies:
\begin{multline}
\label{eqn:corollary result2}
S \leq \underset{i\in \cN}\sum I(X_{i};Y_{i}|Q,U_{i})+ \underset{i \in \cF-\cN}\sum I(X_{i};Y_{i}|Q) \\ + I(U_{\cN}X_{1};Y_{1}|U_{\cF-\cN},Q),\forall \cN \subseteq \cF,
\end{multline}
where  ${\cF} = \{2,3,\hdots,K\}$.
\end{corollary}
\begin{proof}
See Appendix \ref{coroproof} for proof.
\end{proof}

{\em Simple HK schemes:} Consider HK schemes with Gaussian signaling, no timesharing, and no common-private power splitting, i.e., $X_{i} \sim \mathcal{N}(0,P_{i}), \ \ \forall \ 1\leq i\leq K$, $Q$ is constant, and $U_{i} = X_{i}, i \in \cB$  and $U_{i} = \phi, i \notin \cB$ for a fixed ${\cB} \subseteq \{2,3,\hdots,K\}$. The set $\cB$ denotes the indices of the set of transmit messages decoded at receiver 1. For simple HK schemes, we get the following sum rate result directly from Corollary \ref{HKsumm21}.

\begin{corollary}\label{corollary_HK}
The achievable sum rate of a simple HK scheme over the Gaussian many-to-one IC satisfies:
\begin{multline}\label{eqn:corollary_result_1}
S \leq \frac{1}{2}\underset{i\notin \cB}\sum \log (1 + P_{i}) + \frac{1}{2}\underset{i \in \cM}\sum \log (1+P_{i}) \\+ \frac{1}{2} \ \log \left(1 + \frac{P_{1} + \underset{i \in \cB-\cM}\sum h_{i}^2 P_{i}}{1+ \underset{i \notin \cB}\sum h_{i}^2 P_{i} }\right) , \forall \cM\subseteq \cB 
\end{multline}
 for a fixed ${\cB} \subseteq \{2,3,\hdots,K\}$.
\end{corollary}

In the above theorem, we get a sum rate constraint for each subset $\cM$ of $\cB$. 

\subsection{One-to-many IC}
Let $\cI$ denote the set of indices of the receivers at which interference is decoded, and $\cJ$ be the set of receivers at which interference is treated as noise, 
i.e., $\cJ$ = $\{1,2,\cdots ,K-1 \} \backslash \cI $. Let $W_{i}$ be the message at transmitter $i$. The message $W_{K}$ gets split into two parts $W_{K} = \{W_{K0},W_{K1}\} $, where $W_{K0}$ represents the common message that gets decoded at every receiver in $\cI$ and $W_{K1}$ is the private message that gets decoded only at receiver $K$.
\begin{theorem}\label{HKtheoremontotmany}
For the discrete memoryless $K$-user one-to-many IC, the HK achievable rate region is given by the set of all $(R_{1},R_{2},\hdots, R_{K})$ that satisfy 
\begin{IEEEeqnarray}{lcr}
R_{i}\leq I(X_{i};Y_{i}|Q), i\in \cJ \nonumber\\
R_{i}\leq I(X_{i};Y_{i}|Q,U), i \in \cI \nonumber \\
R_{i} + R_{K} \leq I(X_{i},U;Y_{i}|Q) + I(X_{K};Y_{K}|Q,U), i \in \cI \nonumber \\
R_{K}\leq I(X_{K};Y_{K}|Q), \nonumber
\end{IEEEeqnarray}
where $(Q,U,X_{1},X_2,\hdots,X_K)$ is distributed as 
$$p(q,u,x_1,x_2,..x_K) = p(q)\underset{i=1}{\overset{K-1}\prod}p(x_{i}|q)p(u|q)p(x_{K}|u, q). $$
\end{theorem}
\begin{proof}
See Appendix \ref{proofth9}.
\end{proof}

{\em Simple HK scheme:} Let $X_{i} \sim \mathcal{N}(0,P_{i}), \ \ \forall \ 1\leq i\leq K$, $Q$ is constant, and $U = X_{K}$. From Theorem \ref{HKtheoremontotmany}, we directly get the following result.

\begin{corollary}\label{simpleHKonetomany}
The achievable rate region for the simple HK scheme over the Gaussian one-to-many IC is given by:
\begin{IEEEeqnarray}{lcr}
\label{one1}R_{i}\leq \frac{1}{2}\ \log(1+\frac{P_{i}}{1+h_{i}^2P_{K}}), i\in \cJ,\\
\label{e1}R_{i}\leq \frac{1}{2}\ \log(1+P_{i}), i \in \cI,\\
\label{e2}R_{i} + R_{K} \leq \frac{1}{2}\ \log(1 + P_{i}+h_{i}^2P_{K}), i\in \cI,\\
\label{end1}R_{K}\leq \frac{1}{2}\ \log(1+P_{K}).
\end{IEEEeqnarray} 
\end{corollary}

\begin{corollary}\label{sumrate}
The achievable sum rate $S$ for the simple HK scheme over the Gaussian one-to-many IC when $\cJ = \phi$ satisfies
\begin{IEEEeqnarray}{lcr}
\label{srom1}S\leq \underset{j=1}{\overset{K}\sum} \frac{1}{2}\  \log(1+P_{j}), \\
\nonumber S\leq \underset{j\neq i}{\underset{j=1}{\overset{K-1}\sum}} \frac{1}{2}\ \log (1+P_{j}) + \frac{1}{2}\ \log(1+P_{i}+h_{i}^2P_{K}),\\ \forall\ 1\leq i\leq K-1.
\label{next}
\end{IEEEeqnarray}
\end{corollary}
\begin{proof}
See appendix \ref{cor4proof}.
\end{proof}
\section{Sum capacity results }
\subsection{Gaussian many-to-one IC}
           Consider the simple HK scheme with ${\cB} = \{2, 3, \hdots, k \}$, i.e., interference from transmitters 2 to $k$ are decoded at receiver 1. We choose successive indices 2 to $k$ only for notational convenience, and the results can be generalized to any set of $k - 1$ indices by just relabeling the transmitters. For this case, from (\ref{eqn:corollary_result_1}), we have the following $2^{k-1}$ sum rate constraints:
\begin{multline}\label{eqn:corollary_changed_1}
S \leq \frac{1}{2}\underset{i=k+1}{\overset{K}\sum} \log(1 + P_{i}) +\frac{1}{2} \underset{i \in \cM}\sum \log(1+P_{i}) \\+ \frac{1}{2} \ \log\left(1 + \frac{P_{1} + \underset{i \in \cB-\cM}\sum h_{i}^2P_{i}}{1+ \underset{i=k+1}{\overset{K}\sum} h_{i}^2 P_{i} }\right) , \forall \cM\subseteq \cB. 
\end{multline}
The least of these $2^{k-1}$ upper bounds will determine the maximum achievable sum rate for this simple HK scheme. We will now discuss two cases below where we can show that the simple HK scheme achieves sum capacity.

{\em Case 1 }(${\cMI} k_0$): Here we consider the case when the inequality corresponding to $\cM = \cB$ in (\ref{eqn:corollary_changed_1}) is the dominant inequality, i.e., its right hand side is the least. 

\begin{theorem}\label{MIK0Theorem}
For the K-user Gaussian many-to-one IC satisfying the following channel conditions:
\begin{multline}\label{eqn:theorem_result_9}
\underset{i \in \cB-\cN}\prod (1+P_{i}).(1+\underset{j=k+1}{\overset{K}\sum} {h_{j}^2}P_{j} + P_1) \leq \\1+\underset{i\notin \cN}\sum h_{i}^2 P_{i} + P_{1}
, \forall \cN \subset \cB, \cN \ne \cB,
\end{multline}
\begin{equation}\label{MIk0con}
\sum_{j=k+1}^K h_{j}^2 \leq 1,
\end{equation}
where $\cB$ = $\{2,3,\hdots,k\}$ , $k \in \{1,2,..,K\}$, the sum capacity is given by 
\begin{equation}\label{Sumrate_MIk0}
S = \frac{1}{2}\ \log\left(1+\frac{P_{1}}{1+\sum_{j=k+1}^{K}h_{j}^2P_{j}}\right) + \underset{i=2}{\overset{K}\sum}\frac{1}{2}\ \log(1+P_{i}).
\end{equation}
\end{theorem}
\begin{proof}
The converse or upper bound has already been proved in \cite[Thm. 7]{RangaP} under the condition (\ref{MIk0con}) using the genie-aided channel in Fig. \ref{fig:genie1}. This sum rate can be achieved by the simple HK scheme if the inequality corresponding to the $\cM = \cB$ case is the dominant inequality in (\ref{eqn:corollary_changed_1}). This inequality is dominant if the conditions in (\ref{eqn:theorem_result_9}) are satisfied. 
\end{proof}
\begin{remark}
The case of $k = 1$ is taken to be $\cB = \phi$ resulting in condition (\ref{MIk0con}) alone, thereby recovering the sum capacity result for treating all interference as noise in \cite{annvee09}. 
\end{remark}
\begin{remark}
The achievability conditions in (\ref{eqn:theorem_result_9}) are less stringent than the achievability conditions in \cite{RangaP} since joint decoding in the simple HK scheme is better than the successive interference cancellation decoding used in \cite{RangaP}. This can be noted in Fig. \ref{fig1} where the region obtained using this theorem includes an additional shaded region for the case ${\cMI}3_0$ compared to the result in \cite{RangaP}.
\end{remark}

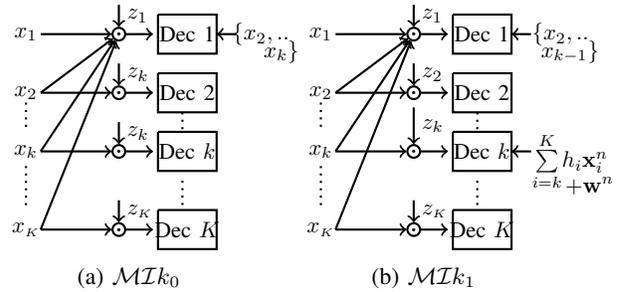
\begin{figure}[h!]
\begin{subfigure}[t]{0.2\textwidth}
\centering
\begin{tikzpicture}[scale=0.13,thick,auto,every node/.style={scale=0.85}]
\draw 
      node at (4,2) [sum, name=suma1] {\dott}
      node at (4,-4) [sum, name=suma2] {\dott}
      node at (4,-10) [sum, name=suma3] {\dott}
      node at (4,-18) [sum, name=suma4] {\dott};
      
\draw[->](-4,2) -- (suma1);
\draw[->](-4,-4) --(suma2);
\draw[->](-4,-10) -- (suma3);
\draw[->](-4,-18) --(suma4);
\draw[->](-4,-4) -- (suma1);
\draw[->](-4,-10) -- (suma1);
\draw[->](-4,-18) -- (suma1);

\draw[->] (4,5) -- node{$z_1$}(suma1);
\draw[->] (4,-1) -- node{$z_k$}(suma2);
\draw[->] (4,-7) -- node{$z_{k}$}(suma3);
\draw[->] (4,-15) -- node{$z_{\scriptscriptstyle{K}}$}(suma4);

\draw[->] (suma1)--(7.5,2) ; 
\draw[->] (suma2)--(7.5,-4) ;
\draw[->] (suma3)--(7.5,-10) ;
\draw[->] (suma4)--(7.5,-18) ;
\draw (8,4)--(14,4)--(14,0)--(8,0)--(8,4);
\draw (8,-2)--(14,-2)--(14,-6)--(8,-6)--(8,-2);
\draw (8,-8)--(14,-8)--(14,-12)--(8,-12)--(8,-8);
\draw (8,-16)--(14,-16)--(14,-20)--(8,-20)--(8,-16);
\draw[->] (16,2)--(14,2);

\draw 
  node at (-5.5,2) {$x_1$}
  node at (-5.5,-4) {$x_2$}
  node at (-5.5,-10) {$x_k$}
  node at (-5.5,-18) {$x_{\scriptscriptstyle{K}}$}
  node at (18.5,2){$\{x_{2},..$} 
  node at (19,0.25){$\ \ \ \ x_{k}\}$}
  node at (11,2){Dec 1}
  node at (11,-4){Dec 2}
  node at (11,-10){Dec $k$}
  node at (11,-18){Dec $K$};
\draw[-,dotted] (10.5,-13) -- (10.5,-15.5);
\draw[-,dotted] (-5.5,-11.5) -- (-5.5,-16);
\draw[-,dotted] (10.5,-6) -- (10.5,-8);
\draw[-,dotted] (-5.5,-5) -- (-5.5,-8);
\end{tikzpicture}
\caption{${\cMI}k_{0}$}
\label{fig:genie1}
\end{subfigure}\ \ \ \ \ 
\begin{subfigure}[t]{0.2\textwidth}
\centering
\begin{tikzpicture}[scale=0.13,thick,auto,every node/.style={scale=0.85}]
\draw 
      node at (4,2) [sum, name=suma1] {\dott}
      node at (4,-4) [sum, name=suma2] {\dott}
      node at (4,-10) [sum, name=suma3] {\dott}
      node at (4,-18) [sum, name=suma4] {\dott};   
\draw[->](-4,2) -- (suma1);
\draw[->](-4,-4) --(suma2);
\draw[->](-4,-10) -- (suma3);
\draw[->](-4,-18) --(suma4);
\draw[->](-4,-4) -- (suma1);
\draw[->](-4,-10) -- (suma1);
\draw[->](-4,-18) -- (suma1);
\draw[->] (4,5) -- node{$z_1$}(suma1);
\draw[->] (4,-1) -- node{$z_2$}(suma2);
\draw[->] (4,-5.5) -- node{$z_{k}$}(suma3);
\draw[->] (4,-15) -- node{$z_{\scriptscriptstyle{K}}$}(suma4);
\draw[->] (suma1)--(7.5,2) ; 
\draw[->] (suma2)--(7.5,-4) ;
\draw[->] (suma3)--(7.5,-10) ;
\draw[->] (suma4)--(7.5,-18) ;
\draw (8,4)--(14,4)--(14,0)--(8,0)--(8,4);
\draw (8,-2)--(14,-2)--(14,-6)--(8,-6)--(8,-2);
\draw (8,-8)--(14,-8)--(14,-12)--(8,-12)--(8,-8);
\draw (8,-16)--(14,-16)--(14,-20)--(8,-20)--(8,-16);
\draw[->] (16,2)--(14,2);
\draw[->] (16,-10)--(14,-10);
\draw 
  node at (-5.5,2) {$x_1$}
  node at (-5.5,-4) {$x_2$}
  node at (-5.5,-10) {$x_k$}
  node at (-5.5,-18) {$x_{\scriptscriptstyle{K}}$}
  node at (19,2){$\{x_{2},..$}
  node at (20,0.25){$x_{k-1}\}$}
  node at (20,-11){$\underset{i=k}{\overset{K}\sum} h_{i}\bx_{i}^n $}
  node at (22,-13.5){$+ \bw^n$}
  node at (11,2){Dec 1}
  node at (11,-4){Dec 2}
  node at (11,-10){Dec $k$}
  node at (11,-18){Dec $K$};
  
\draw[-,dotted] (10.5,-13) -- (10.5,-15.5);
\draw[-,dotted] (-5.5,-11.5) -- (-5.5,-16);
\draw[-,dotted] (10.5,-6) -- (10.5,-8);
\draw[-,dotted] (-5.5,-5) -- (-5.5,-8);
\end{tikzpicture}
\caption{${\cMI}k_{1}$}
\label{fig:genie2}
\end{subfigure}
\caption{Side Information for the genie-aided channels}
\end{figure}

{\em Case 2} (${\cMI} k_1$): Here we consider the case when the inequality corresponding to ${\cM} = {\cB} \backslash \{k\} = \{ 2, 3, \hdots, k-1 \}$ in (\ref{eqn:corollary_changed_1}) is the dominant inequality.

\begin{theorem}\label{Maintheorem}
For the K-user Gaussian many-to-one IC satisfying the following channel conditions:
\begin{IEEEeqnarray}{lcr}
\underset{i\notin \cB - \cN}\prod (1+P_{i})\left( 1+P_{1}+\underset{i=k+1}{\overset{K}\sum} h_{i}^2 P_{i}+\underset{i\in \cB-\cN}\sum h_{i}^2 P_{i}\right) \nonumber\\  
\geq \label{MIK1condn}\prod_{i=2,i\neq k}^{K} (1+P_{i})(1+P_{1}+\sum_{j=k}^{K}h_{j}^2 P_{j})\\
\forall {\cN \subseteq \cB}, {\cN} \neq \{2,3,..,k-1\}\ and\ {\cB} = \{2,3,...k\} \nonumber  \\	
\underset{i=k+1}{\overset{K}\sum} h_{i}^2 \leq 1- \rho^2, \mbox {~~} \rho h_k = 1 + \underset{i=k+1}{\overset{K}\sum}h_{i}^2 P_{i} \label{thm4condn}
\end{IEEEeqnarray}
the sum capacity is given by
\begin{IEEEeqnarray}{lcr}\label{sumcapcase2}
S = \underset{\underset{i\neq k}{i=2}}{\overset{K}\sum} 
\frac{1}{2} \log(1+P_{i}) + \frac{1}{2} \log\Vast(1+ \frac{P_{1}+h_{k}^2P_{k}}
{1+\underset{i=k+1}{\overset{K}\sum}h_{i}^2P_{i}}\Vast). \nonumber
\label{MIK1cap}
\end{IEEEeqnarray}
\end{theorem}
\begin{proof}
  The sum rate $S$ in the theorem statement can be achieved by the simple HK scheme if the inequality corresponding to ${\cM} = {\cB} \backslash \{k\}$ is the dominant inequality in (\ref{eqn:corollary_result_1}). This inequality is dominant if (\ref{MIK1condn}) is satisfied.

For the converse or upper bound, we consider the genie-aided channel in Fig. \ref{fig:genie2}, where a genie provides the signal $\bs_{1}^n=\{\bx_{2}^n,\bx_{3}^n,...,\bx_{k-1}^n\}$ to receiver 1 and the signal $\bs_{k}^n=\underset{i=k}{\overset{K}\sum} h_{i}\bx_{i}^n + \bw^n$ to receiver $k$, where $\bw^n$ is i.i.d. $\mathcal{N}(0,1)$, and $w$ and $z_k$ are jointly Gaussian with $E[w z_k] = \rho$. Now, we have
\begin{IEEEeqnarray*}{lCr}
nS \leq I(\bx_{1}^n;\by_{1}^n|\bs_{1}^n)+\underset{i=2,i\neq k}{\overset{K}\sum}I(\bx_{i}^n;\by_{i}^n) + I(\bx_{k}^n;\by_{k}^n,\bs_{k}^n)\\
= h(\by_{1}^n|\bs_{1}^n) - h(\by_{1}^n|\bs_{1}^n, \bx_{1}^n) + \underset{i=2,i\neq k}{\overset{K}\sum}(h(\by_{i}^n)-h(\bz_{i}^n)) \\ + h\left(\bs_{k}^n\right) + h(\by_{k}^n|\bs_{k}^n) - h(\by_{k}^n,\bs_{k}^n|\bx_{k}^n)\\
\overset{(a)}\leq nh(y_{1G}|s_{1G}) - h(\by_{1}^n|\bs_{1}^n, \bx_{1}^n)
+ \underset{i=2}{\overset{k-1}\sum}(nh(y_{iG})\\-nh(z_{i})) +\underset{i=k+1}{\overset{K}\sum}(h(\by_{i}^n)-nh(z_{i})) 
 + h\left(\bs_{k}^n\right)\\ + nh(y_{kG}|s_{kG}) - h\left(\underset{i=k+1}{\overset{K}\sum} h_{i}\bx_{i}^n + \bw^n|\bz_{k}^n\right) - h(\bz_{k}^n)\\ 
 \overset{(b)}\leq nh(y_{1G}|s_{1G}) 
+ \underset{i=2}{\overset{k-1}\sum}(nh(y_{iG})-nh(z_{i})) \\+\underset{i=k+1}{\overset{K}\sum}(h(\by_{i}^n)-nh(z_{i})) 
 + nh(y_{kG}|s_{kG})\\ - h\left(\underset{i=k+1}{\overset{K}\sum} h_{i}\bx_{i}^n + \bw^n|\bz_{k}^n\right) - nh(z_{k})\\ 
\overset{(c)}\leq nh(y_{1G}|s_{1G})  + \underset{i=2}{\overset{k-1}\sum}(nh(y_{iG})-nh(z_{i})) \\+ \underset{i=k+1}{\overset{K}\sum}(nh(y_{iG})-nh(z_{i}))+ nh(y_{kG}|s_{kG})
\\- nh(\underset{i=k+1}{\overset{K}\sum} h_{i}x_{iG} + w|z_{k})   - nh(z_{k})
\\
\overset{(*)}= nh(y_{1G}|s_{1G}) + \underset{i=2,i\neq k}{\overset{K}\sum}nI(x_{iG};y_{iG}) +\\ nI(x_{kG};y_{kG,}s_{kG})\\
\overset{(d)}=nI(x_{1G};y_{1G}|s_{1G}) + \underset{i=2,i\neq k}{\overset{K}\sum}nI(x_{iG};y_{iG}) +\\ nI(x_{kG};s_{kG})\\
=nI(x_{1G},x_{kG};y_{1G}|s_{1G}) + \sum_{i=2,i\neq k}^{K}nI(x_{iG};y_{iG}),
\end{IEEEeqnarray*}
where $x_{iG} \sim {\cN}(0,P_{i})$, $s_{iG}$ and $y_{iG}$ represent the Gaussian side information and output that result when all the inputs are Gaussian as described in \cite{annvee09}, (a) follows from the fact that Gaussian inputs maximize differential entropy and $h(\by_{k}^n,\bs_{k}^n|\bx_{k}^n) = h(\bz_{k}^n) + h\left(\underset{i=k+1}{\overset{K}\sum} h_{i}\bx_{i}^n + \bw^n|\bz_{k}^n\right)$, (b) follows from $h(\by_{1}^n|\bs_{1}^n, \bx_{1}^n) = h\left(\bs_{k}^n\right)$, (c) follows from application of \cite[Lemma 2]{RangaP}  to $\underset{i=k+1}{\overset{K}\sum}h(\by_{i}^n)  - h(\underset{i=k+1}{\overset{K}\sum} h_{i}\bx_{i}^n + \bw^n|\bz_{k}^n)$ under (\ref{thm4condn}), and (d) follows from the fact that $x_{kG}\rightarrow s_{kG}\rightarrow y_{kG}$ forms a Markov Chain  \cite[Lemma 8]{annvee09} for our choice of $\rho$ in (\ref{thm4condn}).
\end{proof}

The results in Theorems \ref{MIK0Theorem} and \ref{Maintheorem} for the Gaussian $K$-user many-to-one IC are now listed in Table \ref{table:3x3manytoone} for the 3-user case. 
\begin{table}
\renewcommand{\arraystretch}{1.5}
\centering
\begin{tabular}{|p{1cm} | p{6cm} |} 
  \hline
  Strategy & Channel conditions \\ \hline
  ${\cMI}2_{1}$ &  (i) $h_2^2 \leq 1+P_1 + h_{3}^2P_{3},$ \\ 
&$h_3^2  \leq 1 - \left( \frac{1+h_{3}^2P_3}{h_2} \right)^2,h_{2}^2\geq 1$  \\ \cline{2-2}
   & (ii) $h_3^2 \leq 1+P_1 + h_{2}^2P_{2},$ \\ 
&$h_2^2  \leq 1 - \left( \frac{1+h_{2}^2P_2}{h_3} \right)^2,h_{3}^2\geq 1$  \\
  \hline
  ${\cMI}3_{0}$ &  $h_2^2 \geq 1+P_1 , h_3^2 \geq 1+P_1$ \\ 
&$h_{2}^2P_2 + h_{3}^2P_{3} \geq ((1+P_2)(1+P_3)-1)(1+P_{1})$  \\
  \hline
  ${\cMI}3_{1}$ &  (i) $h_2^2 \geq 1+P_1 + h_{3}^2P_{3},h_{3}^2 \leq 1+P_1,h_{3}^2\geq 1,$ \\ 
&$\frac{1+P_{3}}{1+P_2}\ \geq \frac{1+P_1+h_3^2P_3}{1+P_1+h_2^2P_2}$  \\ \cline{2-2}
   & ((i) $h_3^2 \geq 1+P_1 + h_{2}^2P_{2},h_{2}^2 \leq 1+P_1,h_{2}^2\geq 1,$ \\ 
&$\frac{1+P_{2}}{1+P_3} \geq \frac{1+P_1+h_2^2P_2}{1+P_1+h_3^2P_3}$  \\
  \hline
\end{tabular}
\caption{Channel conditions under which sum capacity is achieved using simple HK schemes in Theorems \ref{MIK0Theorem} and \ref{Maintheorem} for the 3-user Gaussian many-to-one IC. Conditions for ${\cMI}1_{0}$ and ${\cMI}2_{0}$ are already given in \cite{RangaP}. These conditions are plotted in Fig. \ref{fig1} for a given set of power constraints.}
\label{table:3x3manytoone}
\end{table}

\begin{remark} 
\label{Tutrem}
In \cite{Tun11}, only a successive decoding strategy where the desired signal is always decoded after decoding the interfering signals, is considered. However, jointly decoding the interfering signal and the desired signal (Scheme ${\cMI} k_1$) is required above to achieve capacity. Furthermore, unlike \cite{Tun11}, the conditions are obtained explicitly in terms of the channel parameters. For more detailed explanation, see Appendix \ref{TunReg}.
\end{remark}
\begin{remark}
\label{Namrem}
The outer bound in \cite[Theorem 2]{Nam15} for the 3-user case matches our outer bound only for ${\cMI} 2_1$. Our $K$-user upper bounds are tighter than the $K$-user upper bounds in  \cite{Nam15arx} for the many-to-one setting. Furthermore, the genie signal used in Theorem \ref{Maintheorem} is different from the genie signals considered in \cite{Nam15arx}. For more detailed explanation, see Appendix \ref{NamReg}.
\end{remark}
\begin{remark}
\label{tserem}
In \cite{rag7}, there is an example 3-user channel where the HK scheme does not achieve capacity, while a scheme based on interference alignment does. It can be verified that this 3-user example channel, when written in standard form, does not satisfy any of the conditions under which sum capacity is derived in this paper. For more detailed explanation, see Appendix \ref{tseReg}.   
\end{remark}

\subsection{Gaussian One-to-many IC}
Consider the simple HK scheme where interference from transmitter $K$ is decoded at $k$ receivers. Without loss of generality, we can consider the set these $k$ receivers to be ${\cI} = \{1,2,\hdots,k\}$ and ${\cJ}=\{k+1,k+2,\hdots,K-1\}$ (other choices can be easily handled by relabeling the receivers). We denote this scheme to be ${\cOI}_{k}$.
\begin{figure}[h!]
\centering
\begin{subfigure}[t]{0.2\textwidth}
\begin{tikzpicture}[scale=0.17,thick,auto,every node/.style={scale=0.85}]
\draw 
      node at (12,0) [sum, name=suma1] {\dott}
      node at (12,-5) [sum, name=suma2] {\dott}
      node at (12,-12) [sum, name=suma3] {\dott}
      node at (12,-17) [sum, name=suma4] {\dott};
\draw[->](5,0) -- (suma1);
\draw[->](5,-5) --(suma2);
\draw[->](5,-12) -- (suma3);
\draw[->](5,-17) -- (suma4);
\draw[->](5,-17) -- (suma1);
\draw[->](5,-17) -- (suma2);
\draw[->](5,-17) -- (suma3);
\draw[->] (12,3) -- node{$z_1$}(suma1);
\draw[->] (12,-2) -- node{$z_k$}(suma2);
\draw[->] (12,-8) -- node{$z_{k+1}$}(suma3);
\draw[->] (12,-14) -- node{$z_{\scriptscriptstyle{K}}$}(suma4);
\draw[->] (19.3,0)--(18.1,0);
\draw[->] (19.3,-5)--(18.1,-5);
\draw[->] (suma1)--(14,0); 
\draw[->] (suma2)--(14,-5);
\draw[->] (suma3)--(14,-12);
\draw[->] (suma4)--(14,-17);
\draw (14,1.1)--(18,1.1)--(18,-1.1)--(14,-1.1)--(14,1.1);
\draw (14,-3.9)--(18,-3.9)--(18,-6.1)--(14,-6.1)--(14,-3.9);
\draw (14,-10.9)--(20,-10.9)--(20,-13.1)--(14,-13.1)--(14,-10.9);
\draw (14,-15.9)--(18,-15.9)--(18,-18.1)--(14,-18.1)--(14,-15.9);
\draw  
  node at(20.5,0) {$x_{K}$}
  node at(20.5,-5) {$x_{K}$}
  node at (3.5,0) {$x_1$}
  node at (3.5,-6) {$x_k$}
  node at (3,-12) {$x_{k+1}$}
  node at (3.5,-17) {$x_{\scriptscriptstyle{K}}$}
  node at (16,0){$Dec_1$}
  node at (16,-5){$Dec_{k}$}
  node at (17,-12){$Dec_{k+1}$}
  node at (16,-17){$Dec_{K}$}
  node at(22.5,-20) {};	
\draw[-,dotted] (16.5,-13) -- (16.5,-15.5);
\draw[-,dotted] (3.5,-13) -- (3.5,-16);
\draw[-,dotted] (20.5,-1) -- (20.5,-5);
\draw[-,dotted] (3.5,-1) -- (3.5,-5);
\end{tikzpicture}
\caption{${\cOI}_{k}$}
\label{fig:genieOI1}
\end{subfigure}\ \ \ \
\begin{subfigure}[t]{0.2\textwidth}
\begin{tikzpicture}[scale=0.17,thick,auto,every node/.style={scale=0.85}]
\draw 
      node at (12,0) [sum, name=suma1] {\dott}
      node at (12,-5) [sum, name=suma2] {\dott}
      node at (12,-12) [sum, name=suma3] {\dott}
      node at (12,-17) [sum, name=suma4] {\dott};
\draw[->](5,0) -- (suma1);
\draw[->](5,-5) --(suma2);
\draw[->](5,-12) -- (suma3);
\draw[->](5,-17) -- (suma4);
\draw[->](5,-17) -- (suma1);
\draw[->](5,-17) -- (suma2);
\draw[->](5,-17) -- (suma3);
\draw[->] (suma1)--(14,0); 
\draw[->] (suma2)--(14,-5);
\draw[->] (suma3)--(14,-12);
\draw[->] (suma4)--(14,-17);
\draw[->] (19.3,-5)--(18.1,-5);
\draw[->] (21.3,-12)--(20.1,-12);
\draw[->] (12,3) -- node{$z_1$}(suma1);
\draw[->] (12,-2) -- node{$z_2$}(suma2);
\draw[->] (12,-9) -- node{$z_{K-1}$}(suma3);
\draw[->] (12,-14) -- node{$z_{\scriptscriptstyle{K}}$}(suma4);
\draw (14,1.1)--(18,1.1)--(18,-1.1)--(14,-1.1)--(14,1.1);
\draw (14,-3.9)--(18,-3.9)--(18,-6.1)--(14,-6.1)--(14,-3.9);
\draw (14,-10.9)--(20,-10.9)--(20,-13.1)--(14,-13.1)--(14,-10.9);
\draw (14,-15.9)--(18,-15.9)--(18,-18.1)--(14,-18.1)--(14,-15.9);
\draw  

  node at(20.7,-5) {$x_{K}$}
  node at(22.7,-12) {$x_{K}$}
  node at (3.5,0) {$x_1$}
  node at (3.5,-6) {$x_2$}
  node at (3,-12) {$x_{K-1}$}
  node at (3.5,-17) {$x_{\scriptscriptstyle{K}}$}
  node at (16,0){$Dec_1$}
  node at (16,-5){$Dec_2$}
  node at (17,-12){$Dec_{K-1}$}
  node at (16,-17){$Dec_{K}$};
\draw[-,dotted] (20.5,-6.5) -- (20.5,-10.5);
\draw[-,dotted] (3.5,-6.5) -- (3.5,-10.5);
\end{tikzpicture}
\caption{${\cOI}_{K-1_{1}}$}
\label{fig:genieOI2}
\end{subfigure}
\caption{Side information for the genie-aided channels}
\end{figure}
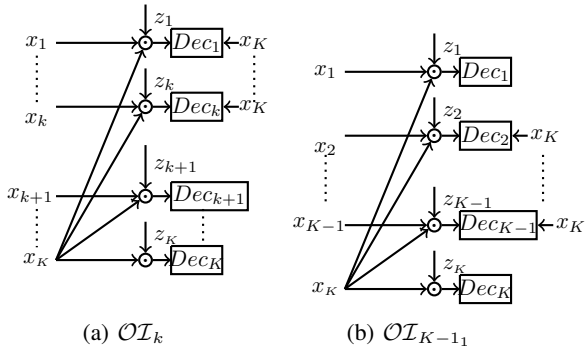

\begin{theorem}\label{One1}
For the $K$-user Gaussian one-to-many IC satisfying the following conditions:
\begin{IEEEeqnarray}{lcr}
1+P_{i} \leq |h_{i}|^2 \label{eqn:condition1} , 1\leq i\leq k, \\
\underset{j=k+1}{\overset{K-1}{\sum}} \frac{|h_{j}|^2P_{\scriptstyle{K}}+|h_{j}|^2}{|h_{j}|^2 P_{\scriptstyle{K}} +1} \leq 1 \label{eqn:condition2},
\end{IEEEeqnarray}
the sum capacity is given by
\begin{IEEEeqnarray}{lcr}
\nonumber S = \frac{1}{2}\ \underset{i=1}{\overset{k}{\sum}} \log(1+P_{i}) +\frac{1}{2}\ \log(1+P_{\scriptstyle{K}})+\\
\label{sum}+\frac{1}{2}\ \underset{j=k+1}{\overset{K-1}{\sum}} \log \left(1+\frac{P_{j}}{1+ |h_{j}|^2 P_{K}}\right). \label{oiksum}
\end{IEEEeqnarray}
\end{theorem}
\begin{proof}
  For achievabililty, consider the achievable rate region in Corollary \ref{simpleHKonetomany} for the simple HK scheme ${\cOI}_k$. Under (\ref{eqn:condition1}), constraint (\ref{e2}) is redundant. From the remaining constraints (\ref{one1}), (\ref{e1}), and (\ref{end1}), we get the achievable sum rate to be equal to the sum capacity in the theorem statement.

  For the converse, consider the genie-aided channel in Fig. \ref{fig:genieOI1}, where a genie provides $x_{\scriptscriptstyle{K}}$ to receivers 1 to $k$. The first $k$ receivers can now achieve the point-to-point channel capacities without any interference. The genie-aided channel can be considered to be a combination of these $k$ point-to-point channels and a Gaussian one-to-many IC with users $k+1$ to $K$ of the original channel. The sum capacity of the $k$ point-to-point channels corresponds to the first term in the right hand side of (\ref{oiksum}). The sum capacity of the Gaussian one-to-many IC with users $k+1$ to $K$ is upper bounded by the sum of the second and third terms in (\ref{oiksum}) under condition (\ref{eqn:condition2}) \cite[Thm. 5]{annvee09}. Thus, we have the required sum capacity result.
  
\end{proof}

Now, we consider the special case where ${\cI} = \{1, 2, \hdots, K - 1 \}$ and ${\cJ} = \phi$, i.e., the interference gets decoded at all receivers. For this special case, we now have a sum capacity result for conditions not included in Theorem \ref{One1}. We will denote this case ${\cOI}_{K-1_1}$.
\begin{theorem}\label{ontoman2}
For the $K$-user Gaussian one-to-many IC satisfying the following conditions:
\begin{IEEEeqnarray}{lcr}
1\leq h_l^2 \leq 1+P_{l} \label{cond1} \\
\frac{h_{l}^2}{1+P_{l}} \leq \frac{h_{i}^2}{1+P_{i}},  1\leq i\leq K-1 \mbox{~and~} i \ne l \label{cond2}
\end{IEEEeqnarray}
for any $l \in \{1, 2, \hdots, K-1 \}$, the sum capacity is 
\begin{equation}
S = \frac{1}{2}\underset{j=1, j \ne l}{\overset{K-1}\sum}\log (1+P_{j}) + \frac{1}{2} \log (1+P_{l} + h_{l}^2 P_{K}).
\label{RATE}
\end{equation}
\end{theorem}
\begin{proof}
  For achievability, consider the achievable sum rate in corollary \ref{sumrate}. The sum capacity in (\ref{RATE}) is the right-hand side of the inequality corresponding to $i = l$ in Corollary \ref{sumrate}. This inequality is the dominant inequality under conditions (\ref{cond1}) and (\ref{cond2}).
  
  For the converse, consider the genie-aided channel (shown in Fig. \ref{fig:genieOI2} for $l = 1$), where a genie provides $x_{\scriptscriptstyle{K}}$ to all receivers 1 to $K-1$ except receiver $l$. The genie-aided channel is a combination of $K-2$ point-to-point channels and a Gaussian one-sided IC with users $l$ and $K$ of the original channel. The sum capacity of the $K-2$ point-to-point channels corresponds to the first term in (\ref{RATE}). The sum capacity of the Gaussian one-sided IC with users $l$ and $K$ is upper bounded by the second term in (\ref{RATE}) under condition (\ref{cond1}) \cite[Thm. 2]{Sas04}. Thus, we have the required result.
  \end{proof}

The results in Theorems \ref{One1} and \ref{ontoman2} for the Gaussian $K$-user one-to-many IC are now listed in Table \ref{table:3x3channelcond} for the 3-user case. 

\begin{table}
\centering
\begin{tabular}{|r | l |} 
  \hline
  Strategy & Channel conditions \\ \hline
  $OI_0$ &  $\underset{j=1}{\overset{2}{\sum}} \frac{h_{j}^2P_{\scriptscriptstyle{K}}+h_{j}^2}{h_{j}^2 P_{\scriptscriptstyle{K}} +1} \leq 1$ \\ \hline
  $OI_1$ &  (i) $h_1^2 \geq 1+P_1, \ 
h_2^2  \leq 1$  \\ \cline{2-2}
   & (ii) $h_2^2 \geq 1+P_1, \
h_1^2  \leq 1$\\
  \hline 
  $OI_2$ & $h_1^2 \geq 1+P_1 ,\
h_2^2 \geq 1+P_2$ \\ 
  \hline
  $OI_{2_{1}}$ &  (i) $1\leq h_2^2 \leq 1+P_1,  
h_2^2  \geq \frac{1+P_2}{1+P_1}h_1^2$  \\ \cline{2-2}
   & (ii) $1\leq h_1^2 \leq 1+P_2,  
h_1^2  \geq \frac{1+P_1}{1+P_2}h_2^2$  \\
  \hline 
\end{tabular}
\caption{Channel conditions under which sum capacity is achieved using simple HK schemes in Theorems \ref{One1} and \ref{ontoman2} for the 3-user Gaussian one-to-many IC. These conditions are plotted in Fig. \ref{fig1} for a given set of power constraints.}
\label{table:3x3channelcond}
\end{table}
\section{Conclusions}
\vspace*{-1mm}
We derived new sum capacity results for the $K$-user Gaussian many-to-one and one-to-many ICs, for new classes of channel conditions (cases ${\cMI}k_{0}$, ${\cMI}k_{1}$, ${\cOI}_k$, ${\cOI}_{K-1_{1}}$). In all these cases, simple HK schemes with Gaussian signaling, no time-sharing and no common-private power splitting achieve sum capacity.
\bibliography{manytoonerefs}
\bibliographystyle{ieeetr}

\appendix
\subsection{Proof of Theorem \ref{HKtheorem}}
\label{RateMotz}
Let $S_i$, $i = 2, 3, \hdots, K$, denote the rates for private messages $W_{i1}$, $i = 2, 3, \hdots, K$, respectively. Let $T_i$, $i = 2, 3, \hdots, K$, denote the rates for common messages $W_{i0}$, $i = 2, 3, \hdots, K$, respectively. Note that $R_i = S_i + T_i$, $i = 2, 3, \hdots, K$. Using standard analysis of HK schemes, we get the following achievable rate region in terms of $\{R_i\}$ and $\{T_i\}$:
\begin{eqnarray}
\label{b4f3}R_{i} - T_{i} \leq I(X_{i};Y_{i}|Q, U_{i})\\
\label{b4f4}R_{i} \leq I(X_{i};Y_{i}|Q),
\end{eqnarray}
for $i = 2, 3, \hdots, K$, and
\begin{eqnarray}\label{b4f1}
R_{1} + \underset{i\in \cN}\sum T_{i} \leq I(U_{\cN},X_{1};Y_{1}|Q,U_{\cF-\cN})
\end{eqnarray}
for all possible ${\cN} \subseteq {\cF}$ and ${\cF} = \{2,3,\hdots,K\}$. We also add the trivial constraints
\begin{equation}
T_i\geq 0, T_i\leq R_i.
\end{equation}

The simplified rate region in (\ref{eqn:theorem_result_1}) and (\ref{eqn:theorem_result_2}) in terms of only the $R_i$'s can be obtained using Fourier-Motzkin elimination. The main steps of the Fourier-Motzkin elimination are provided below.

We eliminate the variables in the following sequence: $T_2, T_3, \hdots, T_K$.
After eliminating $T_2,T_3,\hdots,T_{k}$, the set of inequalities is given by:
\begin{IEEEeqnarray}{lcr}
\nonumber R_{1} + \underset{i \in \cN}\sum R_i + \underset{i\in \cS}\sum T_{i} \leq \underset{i\in \cN}\sum I(X_{j};Y_{j}|Q, U_{j}) \\
\nonumber +I(U_{\cN}, U_{\cS}, X_{1};Y_{1}|U_{\cF-(\cS\bigcup \cN)}, Q), \forall {\cN} \subseteq \{2,3,\hdots, k\}, \\
\nonumber {\cS} \subseteq \{k+1, \hdots, K\}.
\end{IEEEeqnarray} 
For $k+1\leq i \leq K$
\begin{IEEEeqnarray}{lcr}
\nonumber R_i-T_i \leq I(X_i;Y_i|Q, U_i),\\
\nonumber T_{i} \geq 0, T_{i} \leq R_{i}.
\end{IEEEeqnarray}
For $2\leq i \leq K$
\begin{IEEEeqnarray}{lcr}
\label{inductionres}R_{i}\leq I(X_{i};Y_{i}|Q).
\end{IEEEeqnarray}
This can be proved by induction.

Setting $k = K$, we get the required inequalities in (\ref{eqn:theorem_result_1}) and (\ref{eqn:theorem_result_2}) after elimination of $T_2, T_3, \hdots, T_K$.

\subsection{Proof of Corollary \ref{HKsumm21}}
\label{coroproof}
Corollary \ref{HKsumm21} is also proved using Fourier-Motzkin elimination starting from the result in Theorem \ref{HKtheorem}.

First, we substitute $R_1 = S - \sum_{i=2}^K R_i$. Then, we eliminate the variables in the following sequence: $R_2, R_3, \hdots, R_{K}$. After eliminating $R_2,R_3,\hdots,R_{k}$, the set of inequalities is given by:
\begin{IEEEeqnarray}{lcr}
\label{sec}
\nonumber S - \underset{i\in \cB}\sum R_{i} \leq \underset{i\in (\cS-\cB)}\sum I(X_{i};Y_{i}|Q, U_{i}) +\underset{i\in \cN}\sum I(X_{i};Y_{i}|Q, U_{i})\\
\nonumber \underset{i\in \cM-\cN}\sum I(X_{i};Y_{i}|Q) + I(U_{\cS-\cB}, U_{\cN}, X_1;Y_1|U_{\cM-\cN}, U_{\cB}, Q),
\end{IEEEeqnarray} 
$\forall \cB\subseteq \cS$ and ${\cS} = \{k+1,\hdots, K\}$ and $\forall \cN\subseteq \cM$ and ${\cM} = \{2,3, \hdots,k\}$, 
and for $k+1\leq i \leq K$
\begin{IEEEeqnarray}{lcr}
R_{i} \leq I(X_{i};Y_{i}|Q).
\end{IEEEeqnarray}
This can be proved by induction.

Setting $k = K$, we get the required result in (\ref{eqn:corollary result2}) after elimination of $R_2, R_3, \hdots, R_K$.

\subsection{Proof of theorem \ref{HKtheoremontotmany}}\label{proofth9}

Let $S$ denote the rate of the private message $W_{K1}$ and $T$ denote the rate of the common message $W_{K0}$. Note that $R_K = S + T$. Using standard analysis of HK schemes, we get the following achievable rate region in terms of $\{R_i\}$ and $\{T\}$:
\begin{IEEEeqnarray*}{lcr}
R_{i} \leq I(X_{i};Y_{i}|Q), i\in \cJ\\
R_{i} \leq I(X_{i};Y_{i}|U, Q), i \in \cI\\
R_{i} + T \leq I(X_{i}U;Y_{i}|Q), i \in \cI\\
R_K - T \leq I(X_{K};Y_{K}|U, Q)\\
R_K \leq I(X_{K};Y_{K}|Q)
\end{IEEEeqnarray*}
Using Fourier-Motzkin elimination to eliminate $T$, we get the rate region in Thoerem \ref{HKtheoremontotmany}.

\subsection{Proof of corollary \ref{sumrate} }\label{cor4proof}

Given $\cJ = \phi$, we get the following rate constraints:
\begin{IEEEeqnarray*}{lcr}
R_{i}\leq \frac{1}{2}\ \log(1+P_{i}), 1\leq i\leq K-1\\
R_{i} + R_{K} \leq \frac{1}{2}\ \log(1 + P_{i}+h_{i}^2P_{K}),\ 1\leq i\leq K-1\\
R_{K}\leq \frac{1}{2}\ \log(1+P_{K}).
\end{IEEEeqnarray*} 

First, we substitute $R_K = S - \sum_{i=1}^{K-1} R_i$. Then, we eliminate the variables in the following sequence: $R_1, R_2, \hdots, R_{K-1}$. After eliminating $R_1,R_2,\hdots,R_{k}$, the set of inequalities is given by:
\begin{IEEEeqnarray*}{lcr}
R_{i}\leq \frac{1}{2}\ \log(1+P_{i}), k+1\leq i\leq K-1\\
S - \underset{j=k+1}{\overset{K-1}\sum}R_{j} \leq \underset{j=1}{\overset{k}\sum}\frac{1}{2}\ \log(1+P_{j}) + \frac{1}{2}\ \log(1+P_{K})\\
S - \underset{j=k+1}{\overset{K-1}\sum}R_{j} \leq \underset{j=1,j\neq i}{\overset{k}\sum}\frac{1}{2}\ \log(1+P_{j}) +\\ \frac{1}{2}\ \log(1+P_{i}+h_{i}^2P_{K}), 1\leq i\leq k\\
S - \underset{j=k+1,j\neq i}{\overset{K-1}\sum}R_{j} \leq \underset{j=1}{\overset{k}\sum}\frac{1}{2}\ \log(1+P_{j}) + \frac{1}{2}\ \log(1+P_{i}+h_{i}^2P_{K}),\\ k+1\leq i\leq K-1
\end{IEEEeqnarray*}
This can be proved by induction.

Setting $k = K-1$, i.e., after elimination of $R_1, R_2, \hdots, R_{K-1}$, we get the required inequalities in (\ref{srom1}) and (\ref{next}). 

\subsection{Regarding Remark \ref{Tutrem}}
\label{TunReg}
In \cite[Theorem 2]{Tun11}, channel conditions under which sum capacity is achieved for a Z-like Gaussian interference channel are identified. A many-to-one channel can be naturally considered as a special case of this channel. This has been considered in \cite{Tun11} and by considering an appropriate channel matrix $H$ for the many-to-one IC, the channel conditions necessary for achieving sum-capacity for the ${\cMI} 1_0$ case (where all interference is treated as noise) were obtained in \cite[Example 1]{Tun11}. This result cannot be used to obtain the sum capacity results in this paper for any other simple HK scheme.

In \cite[Theorem 3]{Tun11}, the result is extended to a general K-user IC by using a "successive decoding strategy". The simple HK schemes we consider are more general than and include the successive decoding strategy considered in \cite{Tun11}. Furthermore, \cite[Theorem 3]{Tun11} does not give the conditions explicitly in terms of the channel parameters as done in this paper.

If we consider the general (not Z-like) case where only the receiver $k$ suffers interference in the many-to-one IC, the rate conditions from \cite[Theorem 3]{Tun11} are nothing but the HK achievable region when $U_{i} = X_{i}$ for $i \in \{2,3,\hdots k\}$ and $U_i = \phi$ for $i \in \{k+1,\hdots K\}$ in Theorem \ref{HKtheorem}. Following our simplification in this paper gives the constraints on sumrate given by (\ref{eqn:corollary_changed_1}) when $\cM = \cB$ = $\{2,3,...k\}$. Further simplification leads us to get the explicit conditions in terms of the channel paramters for case ${\cMI} k_{0}$ in (\ref{eqn:theorem_result_9}) to achieve sum-rate capacity. Results for the case ${\cMI} k_1$ cannot be obtained based on \cite[Theorem 3]{Tun11} because of the successive decoding strategy limitation.

\subsection{Regarding Remark \ref{Namrem}}
\label{NamReg}
For the 3-user many-to-one IC, the upper bound on sum capacity in \cite[Theorem 2]{Nam15} reduces to
\begin{IEEEeqnarray*}{lcr}
R_1 + R_2 + R_3 \leq I(X_{1G};Y_{1G}) + I(X_{2G};Y_{2G},S_{2G})\\ + I(X_{3G};Y_{3G}),
\end{IEEEeqnarray*}
if $h_{3}^2 \leq \sigma_{V_{N_{2}}}^2$, where $V_{N_{2}} = (N_{2}|Z_{2})$ and $E[Z_{2}N_{2}] = \rho_{N_{2}}$, where $S_{2G} = h_2 X_{2G} + h_3 X_{3G} + N_2$. This upper bound matches with the equation marked (*)  in the proof of Theorem \ref{Maintheorem}. Thus, for the ${\cMI} 2_1$ case, the upper bound in \cite[Theorem 2]{Nam15} and the upper bound in Theorem 4 in this paper match. However, for the other cases and for general $K$, we have a tighter bound for the many-to-one channel. In \cite{Nam15arx}, it is pointed out that \cite[Theorem2]{Nam15} cannot be easily extended to the $K$-user IC.     

\subsection{Regarding Remark \ref{tserem}}
\label{tseReg}
In \cite[Section II.B]{rag7}, the authors consider the following 3 user Gaussian many-to-one IC.
\begin{IEEEeqnarray*}{lcr}
y_1 = \beta \tilde  x_{1} + \beta \tilde x_2 + \beta \tilde x_{3} + z_1,\\
y_2 = \sqrt{\beta} \tilde x_{2} + z_2,\\
y_3 = \sqrt{\beta} \tilde x_{3} + z_3,
\end{IEEEeqnarray*}
where $z_{i} \sim \mathcal{N}(0,1) $ for each $i \in \{1,2,3\}$. The average power constraint at transmitter $i$ is $\tilde P_{i} = 1$ for $i = 1, 2, 3.$ 

This channel can be converted to standard form by defining
\begin{IEEEeqnarray*}{lcr}
x_1 = \beta. \tilde x_1 ,\ 
x_2 = \sqrt{\beta}. \tilde x_2  ,\ 
x_3 = \sqrt{\beta}. \tilde x_3  .
\end{IEEEeqnarray*}
The many-to-one channel in standard form (as in (\ref{m21model1}) and (\ref{m21model2})) is given by
\begin{IEEEeqnarray*}{lcr}
y_1 = x_{1} + \sqrt{\beta} x_2 + \sqrt{\beta} x_{3} + z_1\\
y_2 = x_{2} + z_2\\
y_3 = x_{3} + z_3,
\end{IEEEeqnarray*}
with power constraints
\begin{IEEEeqnarray*}{lcr}
P_{1} = \beta^2  \tilde P_{1} = \beta^2,\ 
P_2 = \beta \tilde P_{2}  = \beta,\ 
P_3 = \beta \tilde P_{3} = \beta,
\end{IEEEeqnarray*}
and $h_2 = h_3 = \sqrt{\beta}$. 

In \cite[Section II. B]{rag7}, the authors prove that for $\beta \geq 2$, the capacity cannot be achieved by any HK-type scheme. We can see that the above channel in standard form does not satisfy any of the conditions given in Table \ref{table:3x3manytoone} or the conditions for ${\cMI}1_{0}$ and ${\cMI}2_{0}$ given in \cite{RangaP}.
As an illustration, we explicitly see how they do not satisfy the conditions necessary for ${\cMI}3_{0}$. The other conditions can be checked similarly. To satisfy the conditions for ${\cMI} 3_0$, the following conditions must be satisfied.
\begin{IEEEeqnarray*}{lcr}
\beta \geq 1 + \beta^2\\
2 \beta^2 \geq ((1+\beta)^2-1)(1+\beta^2)  
\end{IEEEeqnarray*}
We can see that for $\beta \geq 2$ none of the above conditions are satisfied.
\end{document}